\documentclass[moor,sglanonrev]{informs4}
\usepackage{eqndefns-left} 
\RequirePackage{tgtermes}
\RequirePackage{newtxtext}
\RequirePackage{newtxmath}
\RequirePackage{bm}
\RequirePackage{endnotes}

\OneAndAHalfSpacedXII 

\usepackage{hyperref}
\usepackage{algorithm}
\usepackage{algpseudocode}
\usepackage{tikz}
\usepackage{cases}
\usepackage{bm}
\usepackage{thm-restate}
\usepackage{caption}
\usepackage[labelfont=sf]{subcaption}
\theoremstyle{plain}
\newtheorem{lmm}{Lemma}  
\newtheorem{prop}{Proposition}  

\usepackage[sort&compress]{natbib}
 \bibpunct[, ]{[}{]}{,}{n}{}{,}%
 \def\bibfont{\small}%
\EquationsNumberedThrough    

\TheoremsNumberedThrough     
\ECRepeatTheorems  %

\MANUSCRIPTNO{MOOR-0001-2024.00}

\begin{document}
\hypersetup{breaklinks=true}

\RUNAUTHOR{Encz, Mastrolilli, and Vercesi}

\RUNTITLE{B$\&$B-algorithms as PTAS-s}

\TITLE{Branch-and-Bound Algorithms as Polynomial-time Approximation Schemes}

\ARTICLEAUTHORS{%
\AUTHOR{Koppány István Encz}
\AFF{Faculty of Informatics, Università della Svizzera italiana, [CH 6962 Lugano, Switzerland]; Istituto Dalle Molle di studi sull'intelligenza artificiale (IDSIA USI-SUPSI), [CH 6962 Lugano, Switzerland], \EMAIL{enczk@usi.ch}}

\AUTHOR{Monaldo Mastrolilli}
\AFF{University of Applied Sciences and Arts of Southern Switzerland (SUPSI-IDSIA) [CH 6962 Lugano, Switzerland], \EMAIL{monaldo.mastrolilli@supsi.ch}}

\AUTHOR{Eleonora Vercesi}
\AFF{Faculty of Informatics, Università della Svizzera italiana, [CH 6962 Lugano, Switzerland]; Istituto Dalle Molle di studi sull'intelligenza artificiale (IDSIA USI-SUPSI), [CH 6962 Lugano, Switzerland], \EMAIL{eleonora.vercesi@usi.ch}}
} 

\ABSTRACT{%
Branch-and-bound algorithms (B\&B) and polynomial-time approximation schemes (PTAS) are two seemingly distant areas of combinatorial optimization. We intend to (partially) bridge the gap between them while expanding the boundary of theoretical knowledge on the B\&B framework. Branch-and-bound algorithms typically guarantee that an optimal solution is eventually found. However, we show that the standard implementation of branch-and-bound for certain knapsack and scheduling problems also exhibits PTAS-like behavior, yielding increasingly better solutions within polynomial time. Our findings are supported by computational experiments and comparisons with benchmark methods. This paper is an extended version of a paper accepted at ICALP 2025 \cite{encz_et_al:LIPIcs.ICALP.2025.73}.

\MSCCLASS{90C57, 68W25}
\ORMSCLASS{Production/scheduling – Approximations/heuristic; Production/scheduling – Multiple machine; Programming – Algorithms – Branch-and-bound}
}%

\FUNDING{The authors were supported by the Swiss National Science Foundation project n. $200021\_212929$ / 1 "Computational methods for integrality gaps analysis".}





\maketitle


\section{Introduction}\label{sec:intro}

Branch-and-bound algorithms (B\&B) have been a central part of combinatorial optimization for quite some time. They serve as a go-to method for several optimization problems both in theory and in practice, and many $\textsf{NP}$-hard optimization problems are still frequently attacked by a variant of the branch-and-bound method or solvers having it at their core. Their great success could partially be attributed to the generality of the framework, which allows it to be applied to fundamentally different problems; along with the flexibility of choosing parameters such as the branching rule or the search strategy, resulting in an extensive list of options for exploiting the same underlying principles of the framework. For an introduction of the idea and a thorough description, we refer to Kohler and Steiglitz \cite{kohler_steiglitz}; whereas a survey of recent advancements regarding best practices of parameter tuning can be found in \cite{bnb_survey}.

Broadly speaking, the B\&B framework consists of iteratively refining a partition of the search space of a given combinatorial optimization problem. The process is commonly depicted with the help of an auxiliary tree, where the leaves of the tree (often called \emph{active nodes}) correspond to the currently considered partition classes. In this context, refining the partition by further dividing one class is equivalent with creating child nodes for the corresponding leaf in the tree. In addition, each node in the tree has an attributed lower or upper bound (depending on the type of problem in consideration) on the objective value of the best solution in that particular partition class; typically they are calculated from a linear relaxation of the integer formulation of the problem. Occasionally, a node can be discarded from the search process if its attributed bound is worse than an already found solution. 

It is evident from the above description that the framework has various degrees of freedom. The key components that can be adjusted independently are the following:
\begin{itemize}
    \item \textbf{Branching}: The process of dividing a problem into sub-problems.
    \item \textbf{Bounding}: Calculating lower/upper bounds to limit the search space.
    \item \textbf{Selection/Search}: Choosing the next sub-problem (active node) to explore based on some kind of ranking.
\end{itemize}

The choice of these parameters could of course be influenced by the underlying problem. However, when it comes to traversing the branching tree, certain heuristics might be fixed a priori that do not depend on the nature of the problem. Common search heuristics include breadth-first search (BFS), depth-first search (DFS), or best-first search; the latter of which chooses the active subproblem with the best attributed lower/upper bound. 

Several attempts have been made to analyze the effect of preferring one selection strategy to another. For instance, Greenberg and Hegerich \cite{greenberg_hegerich} compare the DFS and the best-first methods applied to the knapsack problem. In addition, Dechter and Pearl \cite{dechter} prove that under some conditions, a variant of the best-first method based on a framework called $A^*$ \cite{hart} visits the least number of nodes before finding the optimum. Moreover, with the increasing interest in applications of artificial intelligence, researchers have even deployed machine learning-based methods that try to \emph{learn} an optimal selection strategy (See, e.g., \cite{alvarez}, \cite{balcan}, \cite{gupta}, \cite{khalil}, \cite{gp_learning}. See \cite{bengio}, \cite{lodi}, \cite{scavuzzo24} for recent surveys). 

However, these comparisons either rely on empirical evidence and rank strategies according to their practical performance, or argue intuitively about the dominance of one method from a certain aspect: DFS is often considered the memory-efficient alternative, whereas the best-first search is often regarded as the ``intuitive'' best choice. To the best of our knowledge, very little effort has been made to justify why one strategy outperforms another, given the problem type.

Meaningful applications of the B\&B framework mostly revolve around $\textsf{NP}$-hard optimization problems, where a standard worst-case analysis does not illuminate the true power of algorithms. Instead, recent results focus on the average-case analysis, and study the performance on randomly sampled instances. Pataki, Tural, and Wong \cite{pataki_basis_2010} analyze the complexity of B\&B for the integer feasibility problem. They show that if the magnitude of the coefficients in the constraint matrix is sufficiently large, then, up to a reformulation technique, almost all instances can be solved directly at the root node. Recently, \cite{dey_branch-and-bound_2023} show that, with any branching rule and best-first as node selection, B\&B reaches the optimum with a polynomial number of nodes on randomly sampled instances, for a fixed number of constraints. Based on this work, Borst et al. \cite{Dadush23, Dadush23-2} show that for some IPs, B\&B trees have size $n^{poly(m)}$ with high probability (i.e., polynomial for fixed $m$), which significantly extends the class of IPs for which B\&B is known to be polynomial.
Some ``negative'' results have also been proposed. Dash \cite{dash_exponential_2005} showed an exponential lower bound on Branch and Cut for 0-1 integer programming when just a few families of cuts are enforced. Dey et al. \cite{Dey22} construct packing, covering and Travelling Salesperson Problem instances for which any complete general B\&B tree must have an exponential size. Dadush et al. \cite{Dadush20} show that any general branch-and-bound tree that proves the integer infeasibility of the cross-polytope has at least $\frac{2^n}{n}$ leaf nodes. Bell and Frieze \cite{bell_solving_2023} show that any B\&B method for the Asymmetric Traveling Salesman Problem via the assignment problem relaxation has an exponential number of nodes.
More recently, \cite{full_strong_branching} showed that for the Vertex Cover problem, choosing full strong branching as the variable selection rule can either perform exceptionally well or be exponentially worse than any other rule, depending on the class of instances considered.

The efficiency of B\&B appears to be strongly dependent on the problem at hand, as well as the choice of lower bound, branching rule, and node selection strategy. This makes it particularly interesting to investigate for which problems, under what conditions, and with which specific choices B\&B can be made to run in polynomial time.

\bigskip

In our current work, we endeavor to explain the intuitive advantage of the best-first search strategy by giving a worst-case theoretical analysis from a slightly unusual point of view. Namely, we will show that for the makespan minimization of unrelated parallel machine scheduling problem with a fixed number of machines (denoted by $Rm||C_{max}$ following standard three-field representation; see \cite{machine_scheduling_review}) and the multiple knapsack problem, the best-first strategy paired with other natural linear programming-based branching and bounding strategies yields a polynomial-time approximation scheme. Thus, practical observations regarding its superiority are strengthened by the guarantee that it is able to find fast a solution \emph{arbitrarily close to the optimum}.

Our contributions are as follows: first (Section \ref{sec:knap}), we consider a family of branch-and-bound algorithms with the best-first tree traversal rule for the multiple knapsack problem, and show that they form a polynomial-time approximation scheme, in which their execution time is constrained to be within polynomial bounds for any fixed approximation ratio. The family $A^{\text{knap}}_{\alpha}$ (parametrized by the approximation parameter $\alpha$) relies on a linear programming-based upper bound and a branching rule that exploits the specific structure of the linear program (see Proposition \ref{2_approx_knapsack}).

\begin{restatable*}{thm}{mk}\label{multi_knapsack_ptas}
    For every fixed $0< \alpha < 1$, the algorithm $A^{\text{knap}}_{\alpha}$ returns an $\alpha$-approximate solution to the multiple knapsack problem, after processing $O(n^{c_{\alpha,m}+1} \cdot m^{c_{\alpha,m}})$-many nodes in the branching tree for some constant $c_{\alpha, m}$ that depends on $\alpha$ and $m$.
\end{restatable*}

Based on the same underlying idea, we provide (Section \ref{sec:js_ptas}) similar results for the makespan minimization of unrelated parallel machine scheduling problem \emph{with a fixed number of machines}, in which we prove that a certain B\&B algorithm is an efficient polynomial-time approximation scheme (EPTAS; for an arbitrary $\varepsilon$ and inputs of size $n$, its running time is $O(n^c \cdot f(1/\varepsilon))$ for some constant $c$ that is independent from $\varepsilon$) for $Rm||C_{max}$. The family $A^{\text{unrel}}_{\varepsilon}$ again relies on a linear programming relaxation and its approximation property. The following theorem  easily implies a polynomial running time for any fixed error.

\begin{restatable*}{thm}{unrel}\label{unrel_machine_ptas}
    For every fixed $\varepsilon > 0$, the algorithm $A^{\text{unrel}}_{\varepsilon}$ returns a $(1+\varepsilon)$-approximate solution to the unrelated machine scheduling problem, after processing at most $m^{\lfloor\frac{m^2}{\varepsilon}\rfloor}$-many nodes in the branching tree.
\end{restatable*}

Exploration and exploitation are two fundamental concepts in search and optimization algorithms. Exploration searches diverse areas, while exploitation refines known good solutions for efficiency. A balance is crucial in algorithms like branch-and-bound to prevent slow convergence. In many branching points of the algorithm, there are decision ambiguities, meaning that (almost) indistinguishable sub-problems keep reappearing in the process. In Section \ref{sec:js:fptas}, we demonstrate that if two nodes represent ``similar'' situations, it is not necessary to explore all such ``similar'' cases to obtain an approximate solution. Given a polynomial bound on the running time, we can put on hold the exploration of some nodes while prioritizing others that are not similar to already explored sub-problems, resulting in an improved exploration of the search space. Standard rounding techniques can be used to identify ``similar nodes''. For the special case of the \emph{uniform machine scheduling} problem with fixed number $m$ of machines (denoted by $Qm||C_{max}$, see \cite{machine_scheduling_review}), this enhances the diversification of the best-first search method, resulting in better exploration of the search space and the achievement of a fully polynomial-time approximation scheme (FPTAS). We will denote the enhanced algorithm by $A^{\text{sim-prof}}_{\varepsilon}$ and prove the following theorem:

\begin{restatable*}{thm}{fptas}\label{fptas}
    For every fixed $\varepsilon > 0$, the algorithm $A^{\text{sim-prof}}_{\varepsilon}$ returns a $(1+\varepsilon)^2$ - approximate solution to the uniform machine scheduling problem, after processing at most $n\cdot \left(\frac{3n(1+\varepsilon)^2}{\varepsilon}\right)^m$ nodes in the branching tree.
\end{restatable*}

Slightly modifying the previous algorithm, we will arrive at a framework that provides a PTAS for the identical machine scheduling problem even when the number of machines $m$ is part of the input (the problem is frequently denoted by $P||C_{max}$).

\begin{restatable*}{thm}{ptasm}\label{ptasm}
    For every fixed $\varepsilon > 0$, the algorithm $A^{\text{eq-prof}}_{\varepsilon}$ returns a $(1+\varepsilon)^2$ - approximate solution to the identical machine scheduling problem with a non-fixed number of machines, after processing at most $\frac{2m^{f(\varepsilon)+1}}{\varepsilon}$ nodes in the branching tree for some function $f$.
\end{restatable*}

To prove these results, we provide structural properties of the vertices of the corresponding polyhedra in Lemma \ref{vertex_char}. We continue with computational experiments supporting our theoretical results in Section \ref{sec:comp_exp}. Finally, with the hope of fueling future research, we list (Section \ref{sec:conc}) the fundamental properties of optimization problems that made our approach applicable.

\emph{Disclaimer:} This paper is an extended version of a paper the authors presented at ICALP 2025 \cite{encz_et_al:LIPIcs.ICALP.2025.73}.

\section{A B\&B PTAS for the Multi-Knapsack Problem}\label{sec:knap}

In the one-dimensional $0-1$ multiple knapsack problem (referred to as \emph{multiple knapsack problem} or \emph{multi-knapsack problem}), we are given $n$ items characterized by weights $\bm{w} = (w_1, \ldots, w_n)$ and  profits $\bm{p} = (p_1, \ldots, p_n)$. Additionally, we have \( m \) knapsacks with capacities \( \bm{C} = (C_1, \ldots, C_m) \), where \( m \) is a fixed constant. The goal is to select a subset of items to maximize the total profit while ensuring that the weight constraints of each knapsack are satisfied. We assume that all weights, profits, and knapsack capacities are non-negative integers.

When $m$ is fixed, which is the case we are considering in this paper, the problem is weakly $\textsf{NP}$-hard and admits an FPTAS \cite{ibarra_kim,lawler_fptas}.
For reference on a comprehensive theory of the problem, we mention the Martello-Toth book \cite{martello_toth}. A survey on recent improvements can be found in \cite{survey,MastrolilliH06}. 

The use of branch-and-bound algorithms for the knapsack problem is well-established. Notable early contributions include those by Kolesar \cite{kolesar}, Greenberg and Hegerich \cite{greenberg_hegerich}, and Horowitz and Sahni \cite{horowitz_sahni_knapsack} for the single-knapsack case. These approaches often leverage various heuristics, such as DFS or BFS search strategies and fractional-item pivot selection rules. However, most of these results focus primarily on computational experiments, with little emphasis on formal theoretical guarantees.

In this work, we demonstrate that a ``standard'' branch-and-bound implementation for the multi-knapsack problem naturally yields a PTAS. Specifically, at each node, the \textbf{bounding} step involves solving the linear programming relaxation known as \emph{surrogate relaxation} (see \cite{martello_toth}, Chapter 6.2.1). 
We then apply the standard \textbf{rounding} technique of Dantzig \cite{dantzig} to obtain an $(m+1)$-approximate feasible solution, and \textbf{branch} according to the most profitable fractional item. The \textbf{selection} strategy is the best-first rule, in which we choose the node to be processed next whose upper bound (the fractional optimum) is the highest; we will denote this strategy by HUB (acronym for Highest Upper Bound) in the experimental section. We terminate whenever the ratio between the global upper bound and the best integer solution reaches or goes above a fixed constant $\alpha\in (0,1)$. The full details of our algorithm, \( A^{\text{knap}}_{\alpha} \), are provided in Section \ref{sec:knap_proof}, along with a proof of the following result:

\mk

For completeness, we sketch a pseudocode of $A^{\text{knap}}_{\alpha}$ in Appendix \ref{app:pseudocode}.

\subsection{Proof of Theorem \ref{multi_knapsack_ptas}}\label{sec:knap_proof}

In this section, we demonstrate that a ``standard'' branch-and-bound implementation for the multi-knapsack problem naturally yields a PTAS. We will need the standard integer programming formulation of the problem:
\begin{equation}\label{knapsack_ip}
    MK_m(\bm{C}, \bm{w}, \bm{p}) : \max \sum\limits_{j=1}^n \sum\limits_{i = 1}^m p_j \cdot x_{j,i} \,\,\, \text{ s.t.}
\end{equation}

\begin{subnumcases}{}
    \sum\limits_{j=1}^n w_j \cdot x_{j,i} \le C_i, & $i \in [m]$, \label{cap_const}\\
    \sum\limits_{i=1}^m x_{j,i} \le 1, & $j \in [n]$, \label{ass_const}\\
    x_{j,i} \in \mathbb{N}, & $j \in [n],\, i \in [m]$, \label{nonneg_const}
  \end{subnumcases}

For the \textbf{Branching} and \textbf{Bounding} components, we will rely on a relaxation of \eqref{knapsack_ip}--\eqref{nonneg_const}. Several such relaxations are discussed in \cite{martello_toth}; the ones that are relevant to us are the \emph{linear programming} relaxation (\eqref{knapsack_ip}, \eqref{cap_const}, \eqref{ass_const}, and non-negativity constraints instead of \eqref{nonneg_const}) and the \emph{surrogate relaxation}, which in essence merges the $m$ knapsacks into one single knapsack with capacity $\sum\limits_{i=1}^m C_i$:

\begin{align}\label{knapsack_surrogate}
    S-MK_m(\bm{C}, \bm{w}, \bm{p})= MK_1 \left(\sum\limits_{i = 1}^m C_i, \bm{w}, \bm{p}\right)
\end{align}

Martello and Toth show in \cite{martello_toth} (Chapter 6.2.1, Theorem 6.2) that the optimum of the linear relaxation of the surrogate relaxation coincides with the optimum of the linear relaxation of the original problem. Using this relationship and the well-known observation of George Dantzig \cite{dantzig}, they describe an algorithm that returns an $(m+1)$-approximate solution. They sort the $n$ items decreasingly by their unit profit $\frac{p_i}{w_i}$, and greedily fill each knapsack in this order until an item no longer fits inside entirely. Then they cut the excessive part and assign it to the next knapsack, and resume the process on this new knapsack with the next item in the queue. Let $\bm{x}^*$ denote this solution to the linear relaxation of \eqref{knapsack_ip}--\eqref{nonneg_const}; by the results of Martello and Toth and Dantzig, $\bm{x}^*$ is optimal. Let $s_1, \ldots, s_m$ denote the (at most) $m$ items that are fractionally assigned in the process. Item $s_k$ is referred to as the \emph{critical item relative to knapsack $k$}, and is obtained as $s_k = \min \left\{j: \sum \limits_{l=1}^j w_l > \sum\limits_{i=1}^k C_i\right\}$.

They conclude that the best of the (at most) $m+1$ integer solutions $\chi_{s_1}, \ldots, \chi_{s_m}, \lfloor \bm{x}^*\rfloor$ has a profit of at least $\frac{1}{m+1}$ times the fractional optimum. Let $\bm{x'}$ denote the most profitable of these $m+1$ assignments. With these notations, we have that
\begin{prop}[Martello, Toth; \cite{martello_toth}]\label{2_approx_knapsack}
    \[
    \bm{p}\cdot \bm{x}' = \max\left\{p_{s_1}, \ldots, p_{s_m}, \sum\limits_{k=1}^{m}\sum\limits_{j = s_{k-1}+1}^{s_k -1} p_j\right\} \ge \frac{1}{m+1}(\bm{p} \cdot \bm{x}^*).
    \]
\end{prop}

In the analysis of the algorithm $A^{\text{knap}}_{\alpha}$, we are going to need the following simple observation, which connects the profit of the best critical item $j^*$ (i.e. the critical item with the highest profit) with the gap of $\bm{x}'$ with respect to the optimal $\bm{x}^*$:

\begin{lmm}\label{knapsack_gap_and_critical_element}
     \[
     \frac{p_{j^*}}{\bm{p}\cdot \bm{x}^*} \ge \min\left\{\frac{1}{m+1}, \,\frac{1}{m} \cdot \left(1-\frac{\bm{p}\cdot \bm{x}'}{\bm{p}\cdot \bm{x}^*}\right)\right\}.
    \]
\end{lmm}

For a fixed $0 <\alpha < 1$,  the specifications of algorithm $A^{\text{knap}}_{\alpha}$ are as follows: as input, we have a triple $(\bm{C}, \bm{w}, \bm{p})$ defining an instance of the multi-knapsack problem. At each step, we select a node $v$ (the \emph{branching node}) among the leaves (the \emph{active nodes}) of a tree we build step-by-step; each node corresponds to a sub-problem in which we fix some variables that are given by the unique path from the root to the node. The \textbf{selection} in our case occurs according to the best-first strategy, where we select the node to be processed next whose attributed upper bound (described later) is the largest of all active nodes. For convenience, let us introduce the ``dummy'' variables $\bar{x}_j = 1-\sum\limits_{i=1}^m x_{j,i}$ for $j = 1, \ldots, n$. Suppose that in the unique path from the root of the tree to $v$, we have fixed $x_{j_1, i_1} = x_{j_2, i_2} = \ldots = x_{j_k, i_k} = 1$, and $\bar{x}_{e_1} = \bar{x}_{e_2} = \ldots = \bar{x}_{e_l}=1$. In other words, items $j_1, \ldots, j_k$ are set to be included in knapsack $i_1, \ldots, i_k$ respectively; whereas items $e_1, \ldots, e_l$ are completely disposed of. Consequently, node $v$ encodes the knapsack sub-problem on the ground set $S=[n]\setminus \left(\bigcup_{z=1}^k \{j_z\} \cup \bigcup_{z=1}^l \{e_z\}\right)$ given by $(\bm{C}_v, \bm{w}_v, \bm{p}_v)$ with $\bm{C}_v = (C'_1, \ldots, C'_m)$ where $C'_i = C_i-\sum\limits_{z:\, i_z = i} w_{j_z}, \,\,\, i \in [m]$; $\bm{w}_v=\bm{w}|_{S}$ and $\bm{p}_v=\bm{p}|_{S}$. 

In the \textbf{bounding} component, a local upper and lower bound $UB(v)$ and $LB(v)$ are determined in the following manner: we consider the appropriate integer program in \eqref{knapsack_ip} with parameters $(\bm{C}_v, \bm{w}_v, \bm{p}_v)$, and its surrogate relaxation. The linear relaxation of \eqref{knapsack_surrogate} is solved to optimality by Dantzig's method giving $\bm{x}^*$, and the $(m+1)$-approximate integer solution $\bm{x}'$ is obtained according to Proposition \ref{2_approx_knapsack}. We save the sub-problem's fractional optimum and the corresponding $(m+1)$-approximate feasible solution into the variables $SUB(v)=\bm{p}_v\cdot \bm{x}^*$ and $SLB(v)=\bm{p}_v\cdot \bm{x}'$; and from these, we create a feasible solution and a local upper bound \emph{to the original problem} by putting items $j_1, \ldots, j_k$ back in knapsacks $i_1, \ldots, i_k$, respectively. Therefore, we set $UB(v) = \bm{p}_v\cdot \bm{x}^* + (p_{j_1}+\ldots p_{j_k})$, and $LB(v)= \bm{p}_v\cdot \bm{x}' + (p_{j_1}+\ldots p_{j_k})$.

Next, we expand the current tree by creating $m+1$ new sub-problems that are going to be represented by the children of $v$; this \textbf{branching} rule is determined by setting the best critical item $j^*$ as the pivot element for this sub-problem. For $i<m+1$, the $i$-th new branch is identified by fixing $x_{j^*, i} = 1$, and corresponds to putting the best critical item $j^*$ in knapsack $i$ while reducing its capacity by $w_{j^*}$. The $(m+1)$-th new branch corresponds to setting $\bar{x}_{j^*} =1$ (and $x_{j^*, 1} = \ldots = x_{j^*, m}=0$ at the same time), and means that we exclude item $j^*$ from all of the knapsacks. We calculate the pertaining upper and lower bounds of all $m+1$ sub-problems. If any of them are infeasible (including the case when item $j^*$ does not fit into knapsack $i$ for some $i$), or their upper bounds are lower than the value of an already found feasible integer solution, we prune the corresponding branch. Otherwise, we add them to the set of active nodes while removing $v$. We update the highest local upper bound ($HUB = \max\{UB(w): w\text{ is active}\}$) and the best integer solution found so far ($HLB = \max\{LB(w): w\text{ is visited}\}$). Finally, we terminate whenever the multiplicative gap between the global upper bound and the best solution found so far reaches or goes above $\alpha$; i.e. $\frac{GLB}{GUB} \ge \alpha$.

Since processing a node $v$ clearly takes polynomial time, a polynomial upper bound on the number of visited nodes in the branching tree means that the above scheme is a PTAS. We sketch a general step of the algorithm in Figure \ref{bnb_sketch}.

\begin{figure}[h]
    \centering
    \includegraphics[width=\linewidth]{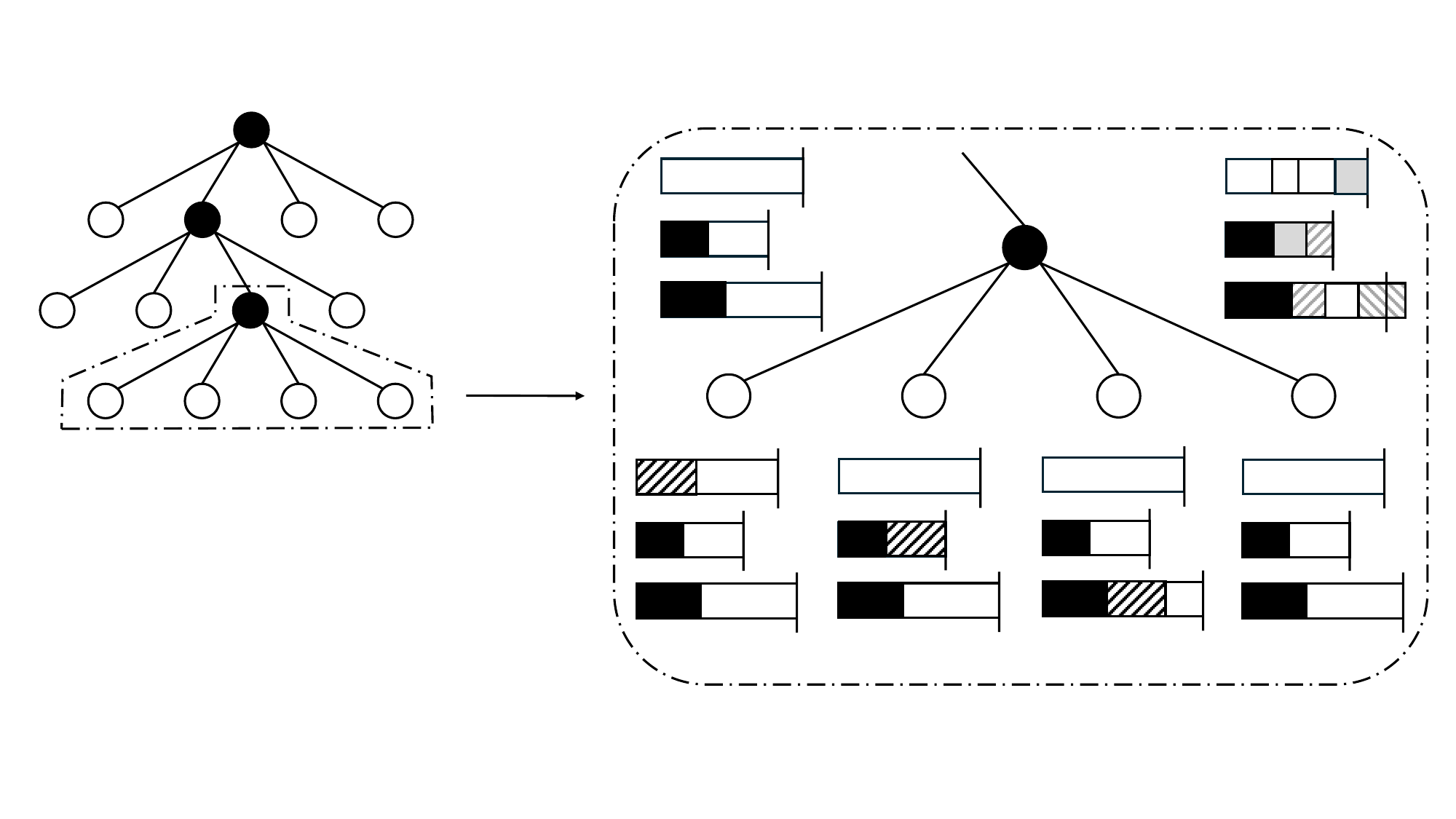}
    \caption{Depiction of a general step of our B\&B for the multi-knapsack problem. On the left is the B\&B tree built so far. White nodes correspond to active nodes. On the right is an example of the branching phase. The node in consideration is the black one. On the top left corner is the corresponding sub-problem, with the black portions of the knapsacks representing the items fixed so far. On the top right, the fractional optimal solution, with $m$ critical items highlighted in grey. The most profitable critical item is fixed on each of the $m+1$ branches, on the $i$-th branch in knapsack $i$ and on the $(m+1)$-th branch it is disposed of.}
    \label{bnb_sketch}
\end{figure}


\begin{proof}{\textbf{Proof of Theorem \ref{multi_knapsack_ptas}}.}
    The proof relies on deriving an upper bound on the number of visited nodes. Let $F = F(\bm{C}, \bm{w}, \bm{p})$ denote the branching tree at termination. On a root-leaf path in $F$, we call an edge \emph{right-turn} if the child node of the edge is the rightmost child of the parent out of the $m+1$ potential children. In other words, the node encodes a step when we leave an item out from all of the knapsacks. Any other step in the path will be called a \emph{left-turn}. The theorem is a simple consequence of the following lemma:

    \begin{lmm}\label{root-leaf}
        There exists a constant $c_{\alpha,m}$ (that depends on $\alpha$ and $m$) such that in every possible tree $F= F(\mathbf{C}, \mathbf{w}, \mathbf{p})$ obtained from an input $(\mathbf{C}, \mathbf{w}, \mathbf{p}$) to the multi-knapsack problem with $m$ knapsacks, every root-leaf path contains at most $c_{\alpha,m}$-many left-turns.
    \end{lmm}

    This lemma indeed implies our theorem, since it implies that the number of different root-leaf paths in $F$ is at most $\binom{n}{c_{\alpha,m}}\cdot m^{c_{\alpha,m}}$, so we can have at most $O(n^{c_{\alpha,m}+1} \cdot m ^{c_{\alpha,m}})$ nodes in $F$.
\end{proof}

\begin{proof}{\textbf{Proof of Lemma \ref{root-leaf}.}}
    For a given $\alpha$ and $n$, let $l(\alpha, n)$ be the maximum number of left-turns in $F$ for any input of $n$ items. Formally,
    \[
    l(\alpha, n) = \sup_{\bm{w}, \bm{p} \in \mathbb{N}^n, \; \bm{C} \in \mathbb{N}^m} \{\text{max. number of left-turns in any path of } F(\bm{C}, \bm{w}, \bm{p})\}
    \]
    Since the maximal number of left-turns is at most $n$ in every possible path in every possible branching tree with $n$ items, we can replace the supremum with maximum. Therefore, we can consider an infinite sequence $(\bm{C}_n, \bm{w}_n, \bm{p}_n)$ and $F_n$ realizing the maximum for each $n$. Our goal now is to prove that $l(\alpha, n)$ remains constant as $n$ tends to infinity. Let us fix $n$, and the corresponding triple $(\bm{C}_n, \bm{w}_n, \bm{p}_n)$ with $F_n$ such that they realize $l(\alpha, n)$. For the sake of simplicity, we will omit $n$ from the notation, and simply consider $(\bm{C}, \bm{w}, \bm{p})$ and $F$.

    Let $F'$ be the inner nodes of $F$ (the tree we get by getting rid of all leaves). Since $|F'| \ge \frac{1}{m+1} \cdot |F|$ (as each group of at most $m+1$ sibling leaves in $F$ has one unique parent in $F'$), it is enough to prove the upper bound on left-turns in $F'$. Consider a root-leaf path in $F'$ with $l(\alpha, n)$-many left-turns, and let the corresponding nodes be $v_1, v_2, \ldots, v_{l(\alpha, n)}$, with items $j_1, \ldots, j_{l(\alpha,n)}$ fixed to be in knapsacks $i_1, \ldots, i_{l(\alpha, n)}$ along the path.

    For the node $v_t$ with $1\le t\le l(\alpha, n)$, let $\bm{x}^* _{v_t}$ and $\bm{x}'_{v_t}$ denote the fractional optimum of the sub-problem's surrogate relaxation and its $(m+1)$-approximate integer rounding, with objective function values $SUB(v_t)$ and $SLB(v_t)$, respectively. Let $\frac{SLB(v_t)}{SUB(v_t)} = r_t$. Recall that $r_t \ge \frac{1}{m+1}$, and that 
    \[
    \frac{LB(v_t)}{UB(v_t)} = \frac{SLB(v_t) + p_{j_1} + \ldots + p_{j_{t-1}}}{SUB(v_t) + p_{j_1} + \ldots + p_{j_{t-1}}}.
    \]
    
    The key observations are the following:
    \begin{itemize}
        \item Since $v_t$ is the inner node of $F$, the algorithm processed it in a previous step and did not halt. Consequently, if $HLB'$ and $HUB'$ denote the back-then global lower and upper bounds, we have that $\alpha > \frac{HLB'}{HUB'}$.
        
        \item $HLB'$ was defined as the best of all integer solutions found up to visiting node $v_t$, so $HLB' \ge LB(v_t)$.
        
        \item By the choice of the tree-traversal strategy being best-first, $UB(v_t) = HUB'$.
    \end{itemize}
    
    Putting these together, we have that
    \[
    \alpha > \frac{HLB'}{HUB'} \ge \frac{LB(v_t)}{UB(v_t)} = \frac{SLB(v_t) + p_{j_1} + \ldots + p_{j_{t-1}}}{SUB(v_t) + p_{j_{1}} + \ldots +p_{j_{t-1}}} = 
    \]
    \[
    =\frac{r_t \cdot SUB(v_t) + p_{j_1} + \ldots + p_{j_{t-1}}}{SUB(v_t) + p_{j_1} + \ldots + p_{j_{t-1}}} = r_t + (1-r_t)\cdot \frac{p_{j_1} + \ldots + p_{j_{t-1}}}{SUB(v_t) + p_{j_1} + \ldots + p_{j_{t-1}}} =
    \]
    \[
    =  r_t + (1-r_t)\cdot \frac{\frac{p_{j_1}}{SUB(v_t)} + \ldots + \frac{p_{j_{t-1}}}{SUB(v_t)}}{1 + \frac{p_{j_1}}{SUB(v_t)} + \ldots + \frac{p_{j_{t-1}}}{SUB(v_t)}}.
    \]
    As an immediate consequence, we have that $\alpha > r_t$, and so
    \begin{equation}\label{inductive_bound}
        1-r_t > 1- \alpha, \,\,\, t = 1, \ldots, l(\alpha, n).
    \end{equation}

    With a universal bound (in $t$) established in \eqref{inductive_bound}, we can now focus on the last left-turn in the path, $v_{l(\alpha, n)}$. Let us introduce a shorthand notation for the fraction on the right-hand side, and repeat what we have for $t=l(\alpha, n)$:
    \[
    f_{l(\alpha, n)} := \frac{\frac{p_{j_1}}{SUB(v_{l(\alpha, n)})} + \ldots + \frac{p_{j_{l(\alpha, n)-1}}}{SUB(v_{l(\alpha, n)})}}
    {1 + \frac{p_{j_1}}{SUB(v_{l(\alpha, n)})} + \ldots + \frac{p_{j_{l(\alpha, n)-1}}}{SUB(v_{l(\alpha, n)})}},
    \]
    and
    \begin{equation}\label{alpha}
        \alpha > r_{l(\alpha, n)} + (1-r_{l(\alpha, n)}) \cdot f_{l(\alpha, n)} = 1 - (1-r_{l(\alpha, n)})(1-f_{l(\alpha,n)}).
    \end{equation}

    \begin{lmm}\label{f}
        \[
        f_{l(\alpha, n)} \ge \frac{(l(\alpha, n) -1)\cdot \min\left\{\frac{1}{m+1},\, \frac{1}{m} \cdot (1-\alpha)\right\}}{1+(l(\alpha, n) -1)\cdot \min\left\{\frac{1}{m+1},\, \frac{1}{m} \cdot (1-\alpha)\right\}}.
        \]
    \end{lmm}

    \begin{proof}{\textbf{Proof.}}
        First, note that
        \[
        SUB(v_{l(\alpha, n)}) \le SUB(v_{l(\alpha, n)-1}) \le \ldots \le SUB(v_2) \le SUB(v_1),
        \]
        since any two consecutive nodes in the sequence $v_i$ and $v_{i+1}$ encode problems where the problem in node $v_{i+1}$ is a sub-problem of the one given by node $v_i$. Second, recall from Lemma \ref{knapsack_gap_and_critical_element} that 
        \[
        \frac{p_{j_t}}{SUB(v_t)} \ge \min\left\{\frac{1}{m+1}, \,\frac{1}{m}\cdot (1-r_t) \right\}, \,\,\, t = 1, \ldots, l(\alpha, n).
        \]
        Last, notice that $f_{l(\alpha, n)}$ is a fraction of the type $\frac{x}{x+1}$, therefore it is monotone increasing in the value of the numerator. Combining these observations with \eqref{inductive_bound}, we can write that
        \[
        f_{l(\alpha, n)} = \frac{\frac{p_{j_1}}{SUB(v_{l(\alpha, n)})} + \ldots + \frac{p_{j_{l(\alpha, n)-1}}}{SUB(v_{l(\alpha, n)})}}
        {1 + \frac{p_{j_1}}{SUB(v_{l(\alpha, n)})} + \ldots + \frac{p_{j_{l(\alpha, n)-1}}}{SUB(v_{l(\alpha, n)})}} \ge 
        \frac{\frac{p_{j_1}}{SUB(v_1)} + \ldots + \frac{p_{j_{l(\alpha, n)-1}}}{SUB(v_{l(\alpha, n)-1})}}
        {1 + \frac{p_{j_1}}{SUB(v_1)} + \ldots + \frac{p_{j_{l(\alpha, n)-1}}}{SUB(v_{l(\alpha, n)-1})}} \ge
        \]

        \[
        \ge \frac{\min\left\{\frac{1}{m+1}, \, \frac{1}{m}\cdot (1-r_1)\right\} + \ldots + \min\left\{\frac{1}{m+1}, \, \frac{1}{m}\cdot (1-r_{l(\alpha,n)-1})\right\}}{1+\min\left\{\frac{1}{m+1}, \, \frac{1}{m}\cdot (1-r_1)\right\} + \ldots + \min\left\{\frac{1}{m+1}, \, \frac{1}{m}\cdot (1-r_{l(\alpha,n)-1})\right\}} \ge 
        \]
        \[
        \ge \frac{(l(\alpha, n) -1)\cdot \min\left\{\frac{1}{m+1},\, \frac{1}{m} \cdot (1-\alpha)\right\}}{1+(l(\alpha, n) -1)\cdot \min\left\{\frac{1}{m+1},\, \frac{1}{m} \cdot (1-\alpha)\right\}}.
        \]
        \hfill \Halmos
    \end{proof}

    Now suppose for contradiction that $\lim\limits_{n \to \infty} l(\alpha,n) = \infty$. Then \eqref{alpha} implies
    \[
    \alpha \ge 1- \lim\limits_{n \to \infty} (1-r_{l(\alpha,n)})(1-f_{l(\alpha, n)}) = 1,
    \]
    since $1 \ge r_{l(\alpha, n)} > 0$ and 
    \[
    \lim\limits_{n \to \infty} f_{l(\alpha, n)} \ge \lim\limits_{n \to \infty} \frac{(l(\alpha, n) -1)\cdot \min\left\{\frac{1}{m+1},\, \frac{1}{m} \cdot (1-\alpha)\right\}}{1+(l(\alpha, n) -1)\cdot \min\left\{\frac{1}{m+1},\, \frac{1}{m} \cdot (1-\alpha)\right\}} = 
    \]
    \[
    =\lim\limits_{x \to \infty} \frac{x}{1+x}=1.
    \]
    However, $1 > \alpha >0$ was assumed. Contradiction. \hfill \Halmos
\end{proof}

\begin{remark}
    With a more careful analysis, we can determine an exact bound on $l(\alpha, n)$ that depends on alpha, just like we do in Theorems \ref{unrel_machine_ptas} and \ref{fptas}. However, the constants in the exponent are not so concise and descriptive as in the other cases, so we decided not to formulate the theorem with the exact values. 
    
    Observe that \eqref{alpha} implies $\alpha > f_{l(\alpha, n)}$, and assume that in Lemma \ref{f}, the minimum is obtained by $\frac{1}{m}\cdot (1-\alpha)$. It follows that 
    \[
    \alpha > \frac{(l(\alpha,n)-1)\cdot \frac{1}{m}\cdot (1-\alpha)}{1+(l(\alpha,n)-1)\cdot \frac{1}{m}\cdot (1-\alpha)} = 1- \frac{1}{1+(l(\alpha,n)-1)\cdot \frac{1}{m}\cdot (1-\alpha)},
    \]
    and 
    \[
    1+(l(\alpha,n)-1)\cdot \frac{1}{m}\cdot (1-\alpha) < \frac{1}{1-\alpha},
    \]
    so
    \[
    l(\alpha,n) < \left(\frac{1}{1-\alpha} -1\right)\cdot \frac{m}{1-\alpha} + 1= \frac{m\alpha}{(1-\alpha)^2}+1.
    \]
    On the other hand, if the minimum in Lemma \ref{f} is achieved by $\frac{1}{m+1}$, a similar analysis shows that 
    \[
    l(\alpha, n) < \frac{m+1}{1-\alpha} + 1.
    \]
    Consequently, we may choose
    \[
    c_{\alpha, m} = 1+\max\left\{\frac{m\alpha}{(1-\alpha)^2}, \frac{m+1}{1-\alpha}\right\}.
    \]
\end{remark}

\section{A B\&B EPTAS for the Unrelated Machine Scheduling Problem}\label{sec:js_ptas}

In the \emph{unrelated parallel machine scheduling} problem, $n$ jobs are assigned to $m$ machines to minimize the \emph{makespan} $\max\{C_S(i) : i = 1, \dots, m\}$, where $C_S(i)$ is the completion time of machine $i$ according to the schedule $S$. Each job $j$ has a machine-dependent processing time $p_{j, i} \in \mathbb{N}$. The problem is strongly $\textsf{NP}$-complete when $m$ is part of the input \cite{garey_johnson_1}, ruling out an FPTAS unless $\textsf{P} = \textsf{NP}$. Furthermore, even a PTAS would imply $\textsf{P} = \textsf{NP}$ \cite{lenstra_shmoys_tardos}. When $m$ is a fixed constant, the problem is denoted by $Rm || C_{\max}$ (following \cite{machine_scheduling_review}), and an FPTAS is possible \cite{horowitz_sahni_job_fptas,grouping_1}. See, e.g. \cite{job_survey} for a survey of more recent advances. Several applications of the B\&B framework were developed for machine scheduling problems, with diverse lower bound strategies ranging from surrogate relaxations \cite{van_de_velde} to lagrangian relaxations \cite{martello_soumis_toth}.  

In this section, we study a B\&B implementation for $Rm || C_{\max}$. At a given node, the corresponding sub-problem is modelled as an integer program. Its LP-based relaxation is solved by a Binary Search (BS) subroutine in the \textbf{bounding} component, followed by a common rounding technique to determine an $(m+1)$-approximate schedule (see \cite{vazirani}, Chapter $17.3$). Then, we \textbf{branch} according to the fractional job that maximizes the minimal processing time $\min \{p_{j, i}: i \in [m]\}$. The \textbf{selection} strategy is again the best-first rule, where the active node with the lowest fractional optimum (denoted by the acronym LLB in the experiments) is picked for processing. We stop whenever the ratio between the global lower bound and the best integer schedule discovered so far reaches or goes below $(1+\varepsilon)$, where $\varepsilon>0$ is a fixed constant. We provide exact details of the algorithm $A^{\text{unrel}}_{\varepsilon}$ in Subsection \ref{sec:js_ptas_proof}, where we show the following:

\unrel

\subsection{Proof of Theorem \ref{unrel_machine_ptas}}\label{sec:js_ptas_proof}

In this section, we consider the problem $Rm||C_{max}$. Our analysis will rely on a \emph{self-similarity} property of the problem; that is, each node of the branching tree must correspond to the same class of LP-formulations. For the sake of exploiting this property, we need that fixing a job to any of the machines should yield another machine scheduling problem. This is not true for the default description, so we introduce the concept of \emph{overheads} denoting the earliest time machines can start completing jobs. Fixing a job $j$ to machine $i$ now corresponds to increasing the overhead of machine $i$ by $p_{j, i}$. Machine $i$ with an overhead $t_i \in \mathbb{N}$ has a completion time $C_S'(i)= t_i + C_S(i)$.

Again, the \textbf{bounding} and \textbf{branching} components will heavily rely on an integer programming formulation of the problem and its linear relaxation. The most straightforward formulation, however, gives rise to some concerns. Instead, we opt to follow in the footsteps of Vazirani \cite{vazirani} and Lenstra, Shmoys, and Tardos \cite{lenstra_shmoys_tardos}. Their proofs rely on a technique called \emph{parametric pruning}, which consists of a binary search for a ``guess'' on the optimal integer makespan while disregarding job-machine pairings that immediately exceed the current guess. For this purpose, they define the following modification of the ``standard'' program: for a given $\bm{P} \in \mathbb{N}^{n\times m}$ and ``guess'' $T \in \mathbb{N}$, let $S_T$ be the set of (job, machine) pairings which do not immediately violate the time limit $T$.
\[
S_T := \{(j,i): p_{j,i} \le T\}.
\]

\begin{definition}
    For $\bm{P} \in \mathbb{N}^{n\times m}, \,\, T \in \mathbb{N}$ and $\bm{t} \in \mathbb{N}^m$, let $PARTIAL-LP-MS_m(\bm{P}, \bm{t}, T)$ be the polyhedron determined by the following set of inequalities:
    \begin{align}\label{partial_machine_scheduling_lp}
        \begin{cases}
            \sum\limits_{i: (j,i) \in S_T} x_{j,i} = 1, & j\in [n], \\
            \sum\limits_{j: (j,i) \in S_T} p_{j,i} \cdot  x_{j,i} \le T-t_i, & i \in [m], \\
            x_{j,i} \ge 0, & (j,i) \in S_T.
        \end{cases}
    \end{align}
    \label{def:binary_search_poly}
\end{definition}

The key properties described in \cite{lenstra_shmoys_tardos} and \cite{vazirani} are easily transcribed to our version with overheads with little to no modification, since they are only dependent on the constraint matrix describing \eqref{partial_machine_scheduling_lp} and not on the right-hand side of the inequalities. Let us recall that for a feasible solution $\bm{x}$, job $j$ is called \emph{fractional} if there exists $i$ such that $x_{j,i}$ does not equal $0$ or $1$ (and therefore has a fractional value); otherwise job $j$ is called \emph{integral}.

\begin{lmm}[Lenstra, Shmoys, Tardos; \cite{lenstra_shmoys_tardos}]\label{rounding}
    If the linear program described in \eqref{partial_machine_scheduling_lp} is feasible, then each vertex $\bm{x}^*$ has at most $m$ fractional jobs. Furthermore, there exists an injection from fractional jobs to the $m$ machines such that each fractional job $j$ is matched to a machine $i$ where $x^*_{j,i}\ne 0$. Moreover, the schedule we get by keeping integral jobs in $\bm{x}^*$ and reassigning fractional jobs to machines according to the injection has a makespan of at most $2T$.
\end{lmm}

Lenstra et al. designed a binary search procedure (starting from an arbitrary integer schedule) for the smallest integer value of $T$ for which the program in \eqref{partial_machine_scheduling_lp} is feasible. They prove that their procedure runs in polynomial time.

\begin{prop}[Lenstra, Shmoys, Tardos; \cite{lenstra_shmoys_tardos}]\label{2_approx_machine}
    Let $T'$ be the result of the binary search; i.e. the smallest integer $T$ for which \eqref{partial_machine_scheduling_lp} is feasible. Furthermore, let $T_{opt}$ be the makespan of the optimal schedule. Then the rounding procedure from Lemma \ref{rounding} applied to a schedule with makespan $T'$ yields an integer schedule with makespan at most $2T_{opt}$.
\end{prop}

For a fixed $\varepsilon > 0$, the specifications of algorithm $A^{\text{unrel}}_{\varepsilon}$ are as follows: as input, we have a matrix $\bm{P} \in \mathbb{N}^{n \times m}$ defining an instance of the unrelated machine scheduling problem. The overhead at the beginning is $\bm{t}\equiv 0$. At each step, we select a node $v$ (the \emph{branching node}) among the leaves (the \emph{active nodes}) of a tree we build step-by-step; each node corresponds to a subproblem in which we fix some job-machine pairings identified by the unique path from the root to the node. The \textbf{selection} in our case occurs according to the best-first selection rule, where we select the node to be processed next whose attributed lower bound (described later) is the smallest of all active nodes. Suppose that in the unique path of length $k$ from the root of the tree to $v$, we have fixed $x_{j_1, i_1} = x_{j_2, i_2} = \ldots = x_{j_k, i_k} = 1$. In other words, job $j_1$ is fixed to machine $i_1$, job $j_2$ is fixed to machine $i_2$, and so on. Consequently, node $v$ encodes the sub-problem with jobs $S=[n]\setminus \bigcup_{z=1}^k \{j_z\}$ given by processing times $\bm{P}_v=\bm{P}|_{S \times [m]}$ and overhead vector determined by the already fixed job-machine pairings: $\bm{t}_v= (t_1, \ldots, t_m)$ with $t_i= \sum\limits_{z: i_z = i} p_{j_z, i_z}, \,\, i \in [m]$. With these, the \textbf{bounding} takes place: a local lower and upper bound $LB(v)$ and $UB(v)$ are determined by applying the binary search of Lenstra et al. to find the smallest integer $T$ (denoted by $T'$) for which \eqref{partial_machine_scheduling_lp} is feasible, and by rounding a vertex of the corresponding polyhedron $PARTIAL-LP-MS_m(\bm{P}_v, \bm{t}_v, T')$ to an integer assignment with makespan at most $(m+1)\cdot T'$. The rounding consists of assigning each fractional job to the machine where its processing time is minimal. We then \textbf{branch} according to the fractional job $j$ whose minimal processing time ($\min\{p_{j,i}: i\in [m]\}$) is maximal. Branch $i$ out of the $m$ new branches fixes job $j$ to machine $i$ and increases its overhead by $p_{j, i}$. 

We calculate the local lower and upper bounds of all $m$ new subproblems. If their lower bounds are greater than the makespan of an already found integer solution, we \textbf{prune} them. Otherwise, we add them to the set of active nodes while removing $v$. We update the global lower and upper bounds $LLB$ (Lowest Lower Bound) and $LUB$ (Lowest Upper Bound): at a given step, they are defined as the minimal local lower bound of the active nodes and the best makespan of an integer solution found so far, respectively. Finally, we terminate whenever the multiplicative gap between the global lower bound and the current champion makespan, $\frac{LUB}{LLB}$, reaches or goes below $1+\varepsilon$.




\proof{\textbf{Proof of Theorem \ref{unrel_machine_ptas}.}}
    Let $F$ be the resulting branching tree, and let $v$ be an arbitrary node different from the one at which the algorithm terminates. Let $LUB'$ and $LLB'$ denote the global upper and lower bounds at the time of processing $v$. By definition, we have $LUB' \le UB(v)$; and by the best-first selection strategy, $LLB'=LB(v)$. The algorithm did not stop after processing $v$; hence
    \[
    1+\varepsilon < \frac{LUB'}{LLB'}\le \frac{UB(v)}{LB(v)}.
    \]
    Let $j^*$ be the fixed job at $v$; i.e. the job among the (at most) $m$ fractional jobs of the vertex whose shortest processing time is maximal. For each job $j'$, let $p'_{j'} = \min\{p_{j',1}, \ldots p_{j', m}\}$ denote the shortest one out of all the $m$ processing times. Suppose that job $j^*$ is the $k$-th according to the decreasing order of $p'_{j'}$-s, and the first $k$ jobs in this order are $j_1, j_2, \ldots, j_{k-1}, j^*$. It is evident that rounding up the fractional jobs cannot increase the makespan by more than $m\cdot p'_{j^*}$, so $UB(v)\le LB(v)+m\cdot p'_{j^*}$.

    Observe that $LB(v) \ge \frac{p'_1 + \ldots + p'_n}{m}$, since the latter is a lower bound on the global fractional optimum, whereas the first is the value of some feasible (fractional) solution to a subproblem. Moreover,
    \[
    p'_1 + \ldots + p'_n \ge p'_{j_1} + \ldots + p'_{j_{k-1}} + p'_{j^*} \ge k \cdot p'_{j^*}.
    \]
    From these, we have that
    \begin{equation}
        1+\varepsilon < \frac{LUB'}{LLB'}\le \frac{UB(v)}{LB(v)} \le 1+\frac{m\cdot p'_{j^*}}{LB(v)} \le1+ \frac{m^2\cdot p'_{j^*}}{k\cdot p'_{j^*}}= 1 + \frac{m^2}{k},
    \end{equation}
    and $k \le \lfloor \frac{m^2}{\varepsilon}\rfloor$. On the path from the root of $F$ to $v$, one job cannot be fixed more than once; so we have that when the depth of $v$ exceeds $\lfloor\frac{m^2}{\varepsilon} \rfloor$, at least one fixed job $j'$ will violate $j' \le \lfloor\frac{m^2}{\varepsilon}\rfloor$. Therefore, the depth of $F$ is at most $\lfloor\frac{m^2}{\varepsilon}\rfloor$ (disregarding the terminating node), and the number of nodes in $F$ is at most $m^{\lfloor\frac{m^2}{\varepsilon}\rfloor}$. Since processing a node takes time that is polynomial in $n$, the overall running time is again polynomial. \hfill \Halmos
\endproof

\section{A B\&B FPTAS for the Uniform Machine Scheduling Problem}\label{sec:js:fptas}

In the \emph{identical machine scheduling} setup, every job takes the same amount of time to complete on each of the machines, and so job $j$ is associated with a single processing time $p_j$. The \emph{uniform machine scheduling} setup extends this by associating a speed $s_i$ with each machine $i \in [m]$, and thus rendering the processing time of job $j$ on machine $i$ to be $p_{j, i} = \frac{p_j}{s_i}$. In this paper, we assume that the vector of processing times $\bm{p}$ and the vector of speeds $\bm{s}$ are such that $p_{j,i} \in \mathbb{N}, \forall j \in [n], \forall i \in [m]$. The problem is frequently denoted by $Qm||C_{max}$ when $m$ is a constant (following \cite{machine_scheduling_review}); it is weakly $\textsf{NP}$-hard and, being a special case of $Rm||C_{max}$, it admits an FPTAS.

In this section, we enhance $A^{\text{unrel}}_{\varepsilon}$ for the special case $Qm||C_{max}$ by exploiting a simple observation. During the course of the algorithm, certain repetitive patterns can be identified based on the jobs that are fixed at a given node. Situations may arise where two or more nodes encode sub-problems that could be classified as similar; provided that the fixed job-machine pairings yield schedules that are sufficiently close to each other (the coordinate-wise distance of schedules required to declare them similar is based on the error tolerance factor $\varepsilon$). When aiming for an approximate solution, the processing of these similar nodes can be delayed indefinitely. When the modified B\&B runs its full course, the returned schedule is either optimal or it is within a range of $(1+\varepsilon)$ of the real optimal solution that we put on hold due to similarity. See \cite{Margot02} for a similar idea applied to a branch-and-cut algorithm.

As we will see, this modified scheme $A^{\text{sim-prof}}_{\varepsilon}$ (which, apart from the enhanced node selection rule, is equivalent with $A^{\text{unrel}}_{\varepsilon}$) reduces the search space by such a great factor that the running time will be polynomial in $1/\varepsilon$ as well.

\fptas

Apart from its increased efficiency, the scheme offers the additional advantage of adaptability to changing requirements and time constraints. If the algorithm completes before its designated time limit or a higher-quality solution is needed, it can resume processing the delayed nodes, further refining the search space using a more precise similarity measure with a smaller $\varepsilon$.

\subsection{Proof of Theorem \ref{fptas}}\label{sec:js_fptas_proof}

Introducing the similarity of profiles helps us narrow down the search space significantly, but in exchange, it causes a slight technical difficulty when analysing our algorithm. More precisely, if we discard nodes from the tree because of similarity, we can no longer guarantee that the global lower bound $LLB$ maintained by the algorithm is a true lower bound on the integer optimum, since the active nodes no longer correspond to a complete enumeration of potential solutions. Therefore, when we reach the stopping criterion $\frac{LUB}{LLB}\le 1+\epsilon$, we might not have a ``real'' approximation guarantee. Luckily, since discarded profiles are very close to some others remaining, we can guarantee that this restriction does not change the global lower bound by too much, and $LLB$ will be at most $\varepsilon$-far from a``real'' lower bound.

Let us fix $1>\varepsilon > 0$, and consider an input $\bm{P} \in \mathbb{N}^{n \times m}$ to the uniform machine scheduling problem with $m$ fixed machines and $n$ jobs. For the sake of a simpler analysis, let us divide each processing time with the global fractional optimum of the ``binary search'' LP-relaxation \eqref{partial_machine_scheduling_lp}. This step does not affect the optimal integer assignment, and its new makespan $T_{opt}$ satisfies that 
\begin{equation}\label{const_norm}
    1 \le T_{opt} \le 2
\end{equation}

These constant bounds help us reduce the space of possible numerical values that emerge at any point of the algorithm. In particular, since we aim for an approximation guarantee of $(1+\varepsilon)^2$, any solution with such makespan guarantee has completion times (and therefore processing times) at most $2(1+\varepsilon)^2$ due to \eqref{const_norm}. Therefore, any time we encounter a sub-problem with completion times or processing times outside the interval $[0,2(1+\varepsilon)^2]$, we can safely disregard it as for sure it cannot correspond to a solution with the desired approximation guarantee.

In this section, we consider the problem $Qm||C_{\text{max}}$, and we will enhance the previous algorithm $A^{\text{unrel}}_{\varepsilon}$ by exploiting common input-modifying techniques that are frequently used for obtaining fully polynomial-time approximation schemes. In general, these techniques consist of applying a series of transformations on the input instance, while keeping the objective value sufficiently close to the optimum. Most often, the modifications are a mixture of rounding down processing times to the nearest value of some finite sequence, and grouping small jobs together to reduce the number of jobs in the input. The rounding of processing times allows for greater control on feasible solutions and gives way to create \emph{profiles} that collect equivalent schedules. On the other hand, grouping small jobs together results in a smaller instance for which even a complete enumeration of schedules would be feasible. If the parameters of the modification are chosen carefully, the combination of these two steps guarantees an algorithm that runs in polynomial time in both $n$ and $\frac{1}{\varepsilon}$. For a detailed background, we refer to \cite{grouping_1} and \cite{grouping_2}.

In our current investigation, we will solely rely on the first type of modification: rounding down processing times to the nearest value of some sequence. But, instead of directly modifying the input, we will design a scheme that allows us to obtain the same effect without touching the input first, thus guaranteeing a more ``natural'' approach. Our method builds on the concept of  
\emph{profiles}: for a (partial) assignment $S$ of some jobs, the profile of $S$ is the $m$-tuple $\Pi(S)=(C_S(1), \ldots, C_S(m))$ of completion times. We call two profiles $\Pi(S_1)$ and $\Pi(S_2)$ $\varepsilon$-similar if $|\Pi(S_1)_i - \Pi(S_2)_i| \le \frac{\varepsilon}{n}, \,\,\,\forall i \in [m]$. We will write $S_1 \sim_{\varepsilon} S_2$ or $\Pi(S_1) \sim_{\varepsilon} \Pi(S_2)$ in notation. The key observation is the following: consider the $m$-dimensional cube $\left[0,2(1+\varepsilon)^2\right]^m$ (the ground set of completion times for which the approximation guarantee is not broken immediately), and consider its partition given by the set of points $\left[0,\frac{\varepsilon}{n}, \frac{2\varepsilon}{n}, \ldots, \frac{2n(1+\varepsilon)^2}{\varepsilon}\cdot \frac{\varepsilon}{n}\right]^m$. If we have two profiles falling into the same partition class, then they are $\varepsilon$-similar. Conversely, any set of profiles with makespan at most $2(1+\varepsilon)^2$ without two $\varepsilon$-similar profiles has at most $\left(1+\frac{2n(1+\varepsilon)^2}{\varepsilon}\right)^m \le \left(\frac{3n(1+\varepsilon)^2}{\varepsilon}\right)^m$ elements (here we implicitly use the fact that $\epsilon < 1$; otherwise the root node contains a $2\le(1+\varepsilon)^2$-approximate solution).

For a node $v$ in the branching tree, its profile $\Pi(v)$ is defined as the profile of the partial schedule $S$ made up of the jobs fixed at $v$. In other words, the profile is simply the overhead vector associated with the integer programming formulation corresponding to the sub-problem at $v$: $\Pi(v) = \bm{t}_v$. The concept of $\varepsilon$-similar profiles allows us to consider nodes of the branching tree ``equivalent'' if they have $\varepsilon$-similar profiles, and they have the same set of jobs fixed so far. Note that we need \emph{both} the same profiles and the same fixed jobs in order to declare two nodes equivalent, as shown by the following identical instance with $2$ machines given by processing time $(n,n,1,\ldots, 1)$ with $n$-many $1$-jobs. We can have two partial assignments with the same profile $(n,n)$, but one of them is made up of one $n$-job and $n$ $1$-jobs while the other is made up of two $n$-jobs. It is not justified to deem them equivalent as the best extension of the first profile has a makespan of $2n$, while the latter can be extended to a schedule with makespan $\frac{3}{2}n$.

However, the following Lemma gives a natural way to ensure that all nodes at a given level have the same fixed jobs in the uniform setup.

\begin{lmm}\label{longest_frac}
   Let $({\bm{P}}, \bm{t}) \in (\mathbb{N}^{(n\times m)}, \mathbb{N}^m)$ be an instance of the uniform machine scheduling problem with $n$ jobs where $\bm{P}$ is given by processing times $\bm{p} \in \mathbb{N}^n$ and machine speeds $\bm{s}\in \mathbb{N}^m$, and let $n$ be the job whose processing time is maximal. Let $T'$ denote the smallest integer $T$ for which $PARTIAL-LP-MS_m (\bm{P}, \bm{t}, T)$ is feasible. If there exists a schedule $\bm{x}^*$ with at least one fractional job such that $\bm{x^*}$ is a vertex of $PARTIAL-LP-MS_m (\bm{P}, \bm{t}, T')$, then there exists a schedule $\bm{\hat{x}}$ in which job $n$ is fractional and $\bm{\hat{x}}$ is a vertex of the same polyhedron.
\end{lmm}

In the proof of Lemma \ref{longest_frac}, we will exploit useful properties of vertices of the polytope described in \eqref{partial_machine_scheduling_lp}. Namely, in the uniform machine scheduling model, we can extend the result of Lemma \ref{rounding} and characterize vertices of the polyhedra. The basic idea behind Lemma \ref{rounding} is the following: for a feasible solution $\bm{x}$, they construct a bipartite auxiliary graph $G(\bm{x})$ with the $2$ classes corresponding to the $m$ machines and the at most $m$ fractional jobs, and add an edge between $ i \in [m]$ and $j \in [n]$ if $x_{j, i} >0$ and is fractional. They conclude that $G(\bm{x})$ must be a \emph{pseudo-forest} when $\bm{x}$ is a vertex, and use this fact to construct a matching between machines and fractional jobs. 

We can strengthen their observation in the uniform model, and use it to our advantage for characterizing vertices of the corresponding polyhedra. Let $(\bm{p}, \bm{s}) \in (\mathbb{N}^{n}, \mathbb{N}^{m})$ denote an input to the uniform machine scheduling problem with $\bm{p}$ being the vector of processing times, and $\bm{s}$ being the vector of machine speeds. The corresponding input matrix is $\bm{P} = (p_{j,i})_{i,j = 1, 1} ^{m, n}$ with $p_{j,i} = \frac{p_j}{s_i}$. Recall that $\bm{P} \in \mathbb{N}^{n \times m}$ is assumed, although it is not explicitly used in the proof. Let us recall that a machine's completion time according to some schedule $S$ is denoted by $C_S'(i)$ when taking into account overhead $t_i$ as well. With a little abuse of notation, a fractional solution $\bm{x}$ of the linear program can be interpreted as a fractional schedule, where the completion time at machine $i$ is denoted by $C_{\bm{x}}'(i)$.

\begin{lmm}\label{vertex_char}
    Let $(\bm{P}, \bm{t})$ be an input to the uniform machine scheduling problem, and let $T'$ denote the smallest integer $T$ for which \eqref{partial_machine_scheduling_lp} is feasible; let $\bm{x}$ be a feasible solution of $PARTIAL-LP-MS_m (\bm{P}, \bm{t}, T')$. Then $\bm{x}$ is a vertex if and only if these two conditions hold: (i) $G(\bm{x})$ is a forest, and (ii) each connected component of $G(\bm{x})$ contains at most one machine-node $i$ for which $C_{\bm{x}}'(i)<T'$.
\end{lmm}

Since the proof of these lemmas are rather tedious, we postpone presenting them to Appendix \ref{app:proofs}.

\bigskip

With this, we are ready to define our final enhanced algorithm $A^{\text{sim-prof}}_{\varepsilon}$. It takes as input an instance of the uniform machine scheduling problem $(\bm{P}, \bm{0})$ where $\bm{P}$ is given by $(\bm{p}, \bm{s}) \in \mathbb{N}^{n+m}$. It rearranges the jobs such that $p_1 \ge \ldots \ge p_n$, then creates an equivalent instance $\bm{P}'$ by dividing $\bm{P}$ with the global fractional binary search-based optimum (at this point, the optimal makespan satisfies $1\le T_{opt}\le 2$). Then it proceeds as a branch-and-bound algorithm with the following specifications: when processing a node $v$, it first finds a vertex of the corresponding relaxation of the sub-problem \eqref{partial_machine_scheduling_lp} with the smallest $T$ for which the program is feasible. If the vertex is integer, the algorithm stops; as it has found a schedule whose makespan coincides with the global lower bound $LLB$ (due to the best-first selection criterion). If the vertex is fractional and the longest unfixed job is not fractional, then it follows Lemma \ref{longest_frac} to arrive at another vertex with the same makespan in which the longest unfixed job is fractional. Then, it rounds up this vertex to find an integer solution, according to an arbitrary matching between machines and fractional jobs. The pivot element at node $v$ will be the longest fractional job, which by now coincides with the longest unfixed job. $m$ new branches are created, labelled by the machine on which the longest unfixed job is fixed at the next level. For each new node $u$, the algorithm checks whether it has already found a schedule with a makespan better than $LB(u)$, in which case $u$ is discarded. Next, $\Pi(u) = \bm{t}_u$ is compared with all previous profiles at the same depth. If $\max\{\Pi(u)_i: i \in [m]\} > 2(1+\varepsilon)^2$, or there already exists a node at the same depth whose profile is $\varepsilon$-similar to $\Pi(u)$, $u$ is discarded. The remaining of the $m$ new nodes are added to the list of active nodes, the list of profiles is appended with the new ones, and the next node to process is selected according to the best-first tree traversal rule.

The process terminates when the ratio between the makespan of the best discovered schedule and the lowest lower bound satisfies $\frac{LUB}{LLB} \le 1+ \varepsilon$, at which point it returns the best schedule found so far.

\begin{proof}{\textbf{Proof of Theorem \ref{fptas}.}}
    The bound on the cardinality of the branching tree is an immediate consequence of the fact that a given level of the tree cannot contain two nodes whose profiles are $\varepsilon$-similar, and we have seen in the beginning of the section that the cardinality of such a set is at most $\left( \frac{3n(1+\varepsilon)^2}{\varepsilon}\right)^m$.

    What remains to be shown is the approximation property of the algorithm. The main challenge lies in the fact that we must coordinate two separate sources of error that yield an approximation guarantee of the multiplied error terms. First, we lose the ``usual'' factor of $(1+\varepsilon)$ as the algorithm stops as soon as it finds a solution within a $(1+\varepsilon)$-radius of $LLB$. On the other hand, as we have hinted before, $LLB$ might not even be a ``true'' lower bound on the integer optimum, since we can no longer guarantee that the set of active nodes cover all possible feasible solutions to the problem. Hence our effort in the following part will be channelled into proving that the B\&B tree restricted to $\varepsilon$-similarity maintains a parameter $LLB$ that is not too far from a ``real'' lower bound.
    
    For that purpose, consider the following three trees: let $F^{\text{sim-prof}}_{\varepsilon}$ denote the branching tree at termination of $A^{\text{sim-prof}}_{\varepsilon}$; let $F^{\text{sim-prof}}$ denote the final branching tree of the variant of $A^{\text{sim-prof}}_{\varepsilon}$ that does not have the early stopping criterion $\frac{LUB}{LLB}\le 1+\varepsilon$; and let $F^{full}$ be the tree which we obtain by running $A^{\text{sim-prof}}_{\varepsilon=0}$ on the same input (in other words, $F^{full}$ is the resulting B\&B tree of an algorithm that considers neither similarity nor early stopping criteria, and finds the optimal solution). It clearly holds that $F^{\text{sim-prof}}_{\varepsilon} \subseteq F^{\text{sim-prof}} \subseteq F^{full}$. Let $F^{\text{sim-prof}}_i$ and $F^{full}_i$ denote the $i$-th levels of the corresponding trees $F^{\text{sim-prof}}$ and $F^{full}$. The core of our proof are the following two lemmas. With a little abuse of definitions, we write $u \sim_{\varepsilon} v$ if $\Pi(u) \sim_{\varepsilon} \Pi(v)$.

    \begin{lmm}\label{i_eps}
        Let $0 < i \le n$ be an arbitrary integer, and let $v\in F^{full}_i$. Then there exists $u \in F^{\text{sim-prof}}_i$ such that $u \sim_{i\varepsilon} v$.
    \end{lmm}
    
    \proof{\textbf{Proof.}}
        We prove by induction on $i$. The case $i=1$ trivially holds, as the root node is included in both $F^{full}$ and $F^{\text{sim-prof}}$. Let $v^{parent}$ be the parent of $v$, and assume that $v$ is the $k$-th child of $v^{parent}$ from the left (equivalently, the algorithm fixed job $i$ on machine $k$ in $v^{parent}$). By induction, there exists a node $w \in F^{\text{sim-prof}}_{i-1}$ such that $v^{parent} \sim_{(i-1)\varepsilon} w$. Let $u$ be the $k$-th child of $w$ in $F^{full}_i$ (the one we get by fixing job $i$ on machine $k$ at node $w$). Notice that since $v^{parent} \sim_{(i-1)\varepsilon} w$, and the coordinate-wise difference between $\Pi(v^{parent})$ and $\Pi(v)$ is the same as between $\Pi(w)$ and $\Pi(u)$; $v \sim_{(i-1)\varepsilon} u$ holds as well. If $u\in F^{\text{sim-prof}}_i$, then we are ready. If not, then there exists $u' \in F^{\text{sim-prof}}_i$ such that $u \sim_{\varepsilon} u'$, so $u' \sim_{i\varepsilon} v$ by the triangle-inequality. \hfill \Halmos
    \endproof

    \begin{lmm}\label{similar_lb}
        Let $u, v \in F^{full}_i$ such that $u \sim_{i\varepsilon} v$. Then $|LB(u)-LB(v)| \le \frac{i\varepsilon}{n}$, and $|UB(u)-UB(v)| \le \frac{i\varepsilon}{n}$.
    \end{lmm}

    \proof{\textbf{Proof.}}
        Let $x^*(u)$ and $x'(u)$ denote the optimal fractional schedule corresponding to the sub-problem at node $u$, and its rounded-up integer complete solution, respectively. Similarly define $x^*(v)$ and $x'(v)$. Let us remind that for an arbitrary integer/fractional schedule $S$, $C_S(i)$ denotes the completion time on machine $i$ according to $S$.

        Suppose that the optimal fractional extension of $u$ increases the completion time of machine $i$ by some value $X_i$: $C_{x^*(u)}(i) = C_u (i) + X_i$. Since $u$ and $v$ consist of the same jobs, extending $v$ by the same fractional job-machine pairings yields a fractional profile whose completion time on machine $i$ is $C_v(i) + X_i$. Assume that the makespan of this fractional schedule is given by the first machine, and so is equal to $C_v (1) + X_1$. Using the fact that $u$ and $v$ are $i\varepsilon$-similar, we have that
        \[
        LB(v) \le C_v(1) + X_1 \le C_u(1) + \frac{i\varepsilon}{n} + X_1 = C_{x^*(u)}(1) + \frac{i\varepsilon}{n} \le LB(u) + \frac{i\varepsilon}{n}. 
        \]
        The other three inequalities can be proven in a similar manner. \hfill \Halmos
    \endproof

    \begin{figure}[h]
        \centering
        \includegraphics[width=\linewidth]{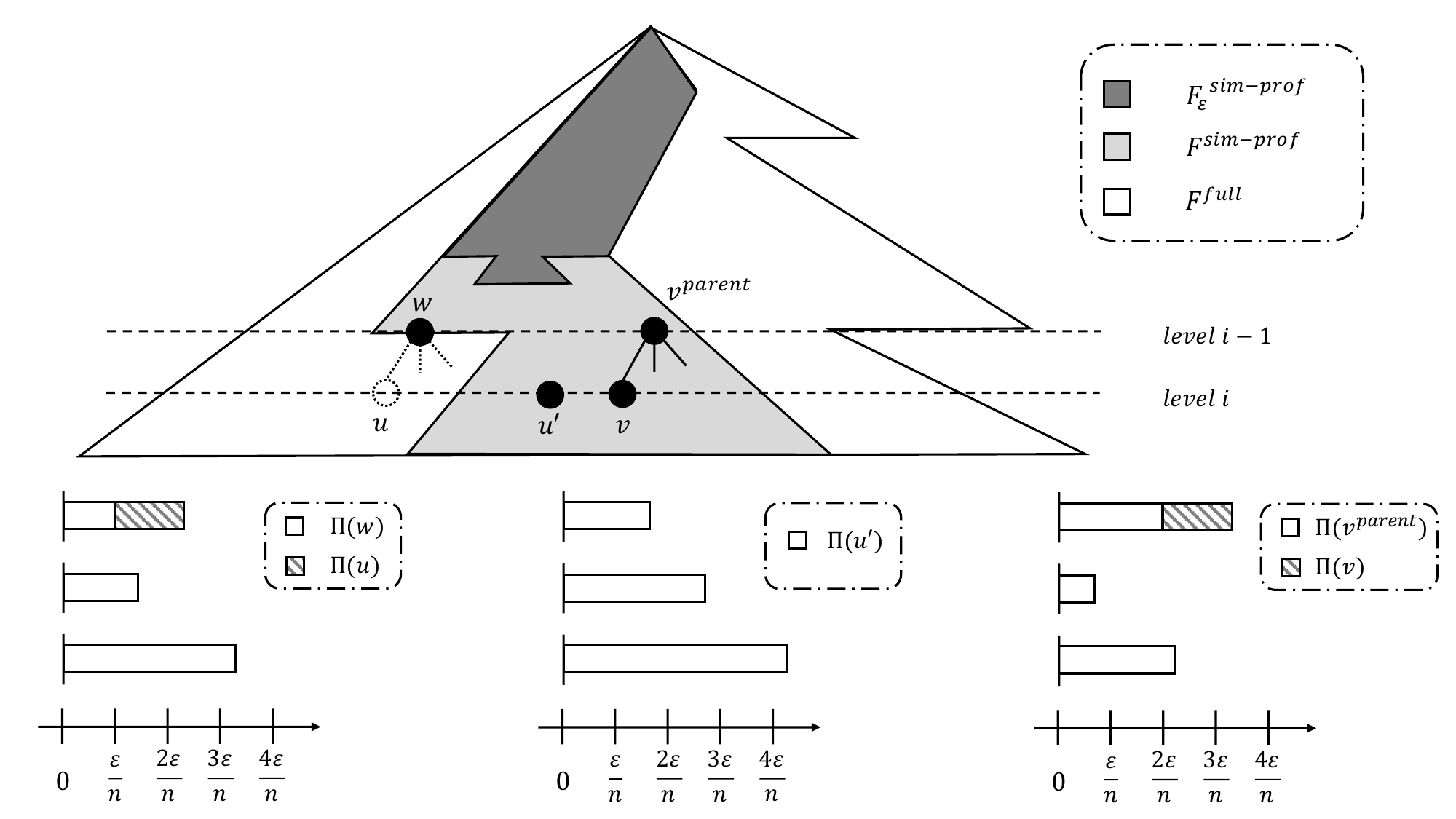}
        \caption{The relationship between the three trees. For $i=2$, $v \sim_{(i-1)\varepsilon} u$, but $u \not \in F^{\text{sim-prof}}$. $u' \sim_{\varepsilon} u$ and $u' \in F^{\text{sim-prof}}$, so $v \sim_{i\varepsilon} u'$.}
        \label{trees}
    \end{figure}
    
    We give an auxiliary sketch in Figure \ref{trees}. With the help of these two lemmas, we can finalize the proof of the theorem. Let $LLB(F)$ and $LUB(F)$ denote the global lower and upper bounds of an arbitrary branching tree $F$, respectively. The relationship between $F^{\text{sim-prof}}_{\varepsilon}$ and $F^{\text{sim-prof}}$ guarantees that $LLB(F^{\text{sim-prof}}_{\varepsilon}) \le LLB(F^{\text{sim-prof}})$. Lemmas \ref{i_eps} and \ref{similar_lb} guarantee that $LLB(F^{\text{sim-prof}})\le LLB(F^{full}) + \varepsilon$, since if $v^{opt}$ denotes the leaf in $F^{full}$ with the lowest lower bound, then the lemmas guarantee the existence of a leaf node $u$ in $F^{\text{sim-prof}}$ such that $LB(u)\le LB(v^{opt})+ n \cdot \frac{\varepsilon}{n} = LB(v^{opt}) + \varepsilon$. Combining these with the stopping criterion of $A^{\text{sim-prof}}_{\varepsilon
    }$, we have that
    \[
    LUB(F^{\text{sim-prof}}_{\varepsilon}) \le (1+\varepsilon)\cdot LLB(F^{\text{sim-prof}}_{\varepsilon}) \le (1+\varepsilon)\cdot LLB(F^{\text{sim-prof}})\le 
    \]
    \[
    \le (1+\varepsilon)\cdot (LLB(F^{full})+ \varepsilon) \le (1+\varepsilon)^2 \cdot LLB(F^{full}),
    \]
    where in the last step we used the fact that $LLB(F^{full}) \ge LB(root) = 1$.
    As $LLB(F^{full})$ is a true lower bound on $T_{opt}$ (as a matter of fact, it coincides with $T_{opt}$), the proof is concluded. \hfill \Halmos
\end{proof}

In what follows we show how the algorithm $A^{\text{sim-prof}}_{\varepsilon}$ generalizes a dynamic programming approach for the machine scheduling problem. The latter consists of constructing a matrix $\bm{M}$, where the $n$ rows are labelled by the $n$ jobs in some fixed order. Column $i$ pertains to a representative profile $\Pi(i)$ of the partition classes of the cube $[0,2(1+\varepsilon)^2]^m$ given by points $[0, \frac{\varepsilon}{n}, \ldots, \frac{2n(1+\varepsilon)^2}{\varepsilon}\cdot \frac{\varepsilon}{n}]^m$. The entry at column $i$ and row $j$ is $1$ if there exists a partial assignment with the first $j$ jobs whose profile is similar to $\Pi(i)$; otherwise, the entry is $0$. The algorithm fills in the entries of the matrix by the best-first principle, then checks all $ 1$-entries of the last row and determines the best makespan of the corresponding profiles. By our above reasoning, the returned schedule will be a $(1+\varepsilon)^2$-approximate solution.

The embedding of the dynamic programming in $A^{\text{sim-prof}}_{\varepsilon}$ can be described as follows: each level of the branching tree $F$ has the same job fixed at every node, therefore level $j$ contains nodes where the fixed jobs are $1,\ldots, j$. Moreover, we only prune a node when either its profile is similar to another one already found (implying that at most one profile is considered from each partition class), or its lower bound is worse than an already found integer solution. Therefore, nodes at depth $j$ in $F$ have a one-to-one correspondence with cells of the $j$-th row of $M$ whose entry is $1$, except for some profiles that were discarded for having a too-high lower bound. In other words, $A^{\text{sim-prof}}_{\varepsilon}$ can be seen as the dynamic programming algorithm embedded in the branch-and-bound framework, where $\bm{M}$ is traversed according to the best-first logic, and some entries are disregarded when even the best possible extension of their profile is worse than an already found feasible solution. In the worst case, each cell of $\bm{M}$ is visited before a $(1+\varepsilon)^2$-approximate solution is found; but it happens no later than the processing of the last entry, according to Theorem \ref{fptas}.

The embedding becomes even more evident if we replace the best-first node selection rule with BFS. Then, traversing level $j$ of the tree is nothing else but processing the $j$-th row of $\bm{M}$ except for some entries that are stepped over because of their lower bound.

\bigskip

To conclude the section, we point out the infeasibility of repeating our results for the unrelated machine scheduling problem and the multiple knapsack problem. The notion of profiles and the $\frac{\varepsilon}{n}$-partition of their space can be extended without changing anything. The difficulty lies in guaranteeing the highly structured property of the branching tree, in which the same job is fixed at all nodes of a given level. Of course, one can simply hard-code this into the algorithm, but giving a pivot rule that achieves this naturally seems infeasible. In particular, the following example shows that the ``maximal shortest processing time'' selection rule does not have this guarantee: let $m = 3$ and consider the following input with $n=2k+2$ jobs: $p_{1,1} = p_{1,2}=p_{1,3} = 3k+2$, $p_{j,2} = p_{j,3} = 3, \,\,\, j= 2, \ldots, n-1$ and $p_{n,2} = p_{n,3} = 2$. The rest of the processing times are chosen such that $p_{j,1} \le 3k+1, \,\,\, j = 2, \ldots, n$. It is easy to check that in the second iteration, there is no vertex in the polyhedron where the job with the maximal shortest processing time is fractional. This instance also serves as a counterexample for a bunch of other pivot selection strategies, such as ``maximal average completion time'' or ``maximal longest processing time''.

A similar phenomenon takes place in the case of the (multiple) knapsack problem, with the exception that we \emph{know} the infeasibility of having a structure where each node in the same level has the same job fixed. In particular, for the single knapsack problem, the two children of a given node have different pivot elements (provided that they are both feasible).

\begin{lmm}
    Let $(C, \bm{w}, \bm{p})$ denote an input to the single knapsack problem. Assume that the items are such that $\frac{p_1}{w_1} > \ldots > \frac{p_n}{w_n}$, and the pivot element is $j^*$. Let $(C-w_{j^*}, \bm{w}', \bm{p}')$ and $(C, \bm{w}', \bm{p}')$ denote the two subproblems corresponding to including and excluding item $j^*$ from the knapsack, with $\bm{w}' = \bm{w}|_{[n]-j^*}$ and $\bm{p}' = \bm{p}|_{[n]-j^*}$. Assume that both subproblems are feasible, and the corresponding pivot elements are $j_1$ and $j_2$. Then $j_1 < j^* < j_2$.
\end{lmm}

\proof{\textbf{Proof.}}
    Items $1, \ldots, j^* -1$ do not fit inside the knapsack with reduced capacity $C-w_{j^*}$ (because items $1, \ldots, j^*-1, j^*$ did not fit inside the knapsack with original capacity $C$), but they do fit inside the original capacity $C$. Hence, $j_1 < j^* < j_2$. \hfill \Halmos
\endproof

\section{A B\&B PTAS for the Identical Machine Scheduling Problem with a non-constant number of machines}\label{sec:js_ptas_m}

For all previous algorithms, the number of machines, $m$, appeared in the exponent of the running time. In this section, we are going to explore how a refined analysis of $A^{\text{sim-prof}}_{\varepsilon}$ in the identical machines paradigm helps us achieve a time complexity that is polynomial in $m$ as well; thus obtaining a PTAS for the identical machine scheduling problem with an unfixed number of machines (denoted by $P||C_{max}$). The key observation is that when machines are identical, the algorithm does not have to take into account the order of the machines when considering the similarity of profiles. We will slightly adapt the notion of profiles to mirror this observation, along with a collection of other modifications:

\begin{itemize}

    \item Instead of an additive partitioning of the interval $[0,2(1+\varepsilon)^2]$, we opt to use one based on a geometric sequence of partition points of the following form: $\varepsilon \cdot (1+\varepsilon)^h, \,\,\, h = 0, \ldots , \log_{1+\varepsilon} \frac{2(1+\varepsilon)^2}{\varepsilon}$. The reason is that with this partition, the number of classes is constant in both $m$ and $n$; and the incurred error when considering similarity \emph{for all jobs simultaneously} is just a multiplicative factor of $(1+\varepsilon)$, since if all jobs in the profile get multiplied by $(1+\varepsilon)$, then so does the makespan of the entire profile. If we want to get the same error term with an additive partitioning, we would need that the distance between consecutive points must be at most $\frac{\varepsilon}{n}$, but then the number of partition classes ($\ge \frac{2n(1+\varepsilon)^2}{\varepsilon}$) would not be constant in the number of jobs.

    \item We will rely on a distinction between small jobs and big jobs that is relatively common in the machine scheduling literature (\cite{grouping_1}, \cite{JansenM04}, \cite{grouping_2}, \cite{Mastrolilli03}, \cite{Mastrolilli04}, \cite{Mastrolilli06}, \cite{MastrolilliH06}). Jobs with processing time less than $\varepsilon$ will not be considered for determining possible profiles. The reason is that profiles containing jobs with processing times less than $\varepsilon$ will never be considered for branching. By the way the algorithm performs the creation of branches, fixing a job $j$ with $p_j < \varepsilon$ can only occur when the longest fractional job in the fractional optimum of the current sub-problem was shorter than $\varepsilon$, but then rounding up the fractional jobs according to Lenstra-Shmoys-Tardos method yields a feasible solution with makespan at most $\varepsilon$ longer than the fractional optimum. If the tree traversal rule was best-first, it means that at the time of processing the node, the integer solution differed by at most $\varepsilon$ from the global lower bound, and therefore it is a $(1+\varepsilon)^2$-approximate solution. From this point onward, a job $j$ will be labelled \emph{small} if $p_j < \varepsilon$ and \emph{big} otherwise.

    \item We slightly refine the definition of similarity, as the previous one had an important flaw: in a sequence of points where each consecutive pair is $\varepsilon$-similar, the first and last points might not be $\varepsilon$-similar at all, and making decisions based on the sequence might propagate the error on an exponential scale. Therefore, we consider a more restrictive notion of similarity that we describe now.
\end{itemize}

For any number $\varepsilon \le x \le 2(1+\varepsilon)^2$, let $round(x)=\sup\{\varepsilon\cdot (1+\varepsilon)^k: k\in \mathbb{N}, \varepsilon\cdot (1+\varepsilon)^k \le x\}$. For a partial schedule $S$ without small jobs, let $round(S)$ denote the partial schedule we get by replacing, for each job $j$ appearing in $S$, the processing time $p_j$ by $round(p_j)$. We want to count the number of essentially different completion time vectors that can be obtained by the $round$ operator, bearing in mind that the order of machines does not matter due to the identical framework. The main thing to note is that for any partial schedule $S$ with makespan at most $2(1+\varepsilon)^2$, each completion time in $round(S)$ can take at most constant-many number of values.

\begin{lmm}\label{varepsilon-grid}
    Let $S$ be a partial schedule with makespan at most $2(1+\varepsilon)^2$ that is composed of only big jobs. Then the number of potential values of a completion time in $round(S)$ is at most $f(\varepsilon)= 8\cdot \left(\frac{1}{\varepsilon}\right)^{\log_{1+\varepsilon} \frac{2(1+\varepsilon)^2}{\varepsilon}}$.
\end{lmm}

\begin{proof}{\textbf{Proof.}}
    A completion time in $round(S)$ is given as $\varepsilon\cdot (1+\varepsilon)^{h_1} + \ldots \varepsilon\cdot (1+\varepsilon)^{h_k}$ for some $k, h_1, \ldots, h_k\in \mathbb{N}$. Consider this expression as a formal power series in $\varepsilon$ over non-negative real numbers; in other words, an element of $\mathbb{R}_{\ge 0}[\varepsilon]$. Using the binomial theorem $(a+b)^n = \sum\limits_{i=0}^n \binom{n}{i}a^i b^{n-i}$, and grouping elements in the expansion having the same exponent of $\varepsilon$, we can rewrite the expression in the form $r_0 \cdot 1 + r_1 \cdot \varepsilon + \ldots + r_l \cdot \varepsilon^l$. While there exists an $i>0$ such that $r_i \ge \frac{1}{\varepsilon}$, we can simplify the expression by setting $r_i:=  r_i - \frac{1}{\varepsilon}$ and $r_{i-1}:=  r_{i-1}+1$. Let us call the final expression, in which no such simplifying step exists, the \emph{reduced form} of the expression. In the reduced form, each $r_i$ is at most $\frac{1}{\varepsilon}-1$ for $i>0$, and $r_0 \le 8$ since the makespan of $round(S)$ is at most $2(1+\varepsilon)^2\le 8$ (here we assume $\varepsilon \le 1$, otherwise the root node contains a $(1+\varepsilon)$-approximate solution). Moreover, the largest exponent of $\varepsilon$ (denoted by $l$ in the reduced form) is bounded by $\max\{h_1, \ldots, k_k\}$, which is at most $\log_{1+\varepsilon} \frac{2(1+\varepsilon)^2}{\varepsilon}$ due to the makespan of $round(S)$ being at most $2(1+\epsilon)^2$. Hence, the number of reduced forms, and therefore the number of potential completion times of a machine in $round(S)$, is bounded by $8\cdot (1/\varepsilon)^{\log_{1+\varepsilon} \frac{2(1+\varepsilon)^2}{\varepsilon}}$. \hfill \Halmos
\end{proof}

Now we are ready to provide the new definition of profiles. According to Lemma $\ref{varepsilon-grid}$, the number of different completion times in $round(S)$ is some constant $f(\varepsilon)$; let us consider an arbitrary enumeration of them (if the constant in Lemma \ref{varepsilon-grid} is not integer, we can take its integer part as a valid upper bound as well, since the number of potential values is always integer). For $i\in [f(\varepsilon)]$, let $c_i(S)$ denote the number of machines on which $round(S)$ has a completion time equal to the $i$-th potential completion time in the fixed order. Note that $\sum\limits_{i=0}^{f(\varepsilon)} c_i(S) = m$. We define the new profile vector as $\Pi_{\text{new}}(S)=(c_1(S), \ldots, c_{f(\varepsilon)}(S))$. The value of each $c_i (S)$ can be at most $m$; hence, the new profile vector can take at most $m^{f(\varepsilon)}$-many values. We call two profiles/partial schedules $\varepsilon$-equivalent if $\Pi_{\text{new}}(S_1) = \Pi_{\text{new}}(S_2)$ coordinate-wise. We will write $\Pi_{new}(S_1) \simeq_{\varepsilon} \Pi_{new}(S_2)$ or $S_1 \simeq_{\varepsilon} S_2$ in notation. Note that as opposed to $\varepsilon$-similarity, $\varepsilon$-equivalency is now an equivalence-relation.

The idea behind $\varepsilon$-equivalent profiles is almost identical to $\varepsilon$-similar profiles: two partial schedules with equivalent profiles almost always have makespans within a factor of $(1+\varepsilon)$ to each other, except maybe when the partial schedules contain a job with processing time less than $\varepsilon$. However, as we have discussed before, we make sure that this will not compromise the framework by showing that as soon as a small job is processed, the algorithm stops. Moreover, due to the machines being identical, if the partial schedules are made up of the exact same jobs, their optimal extensions to complete schedules can also differ by at most a multiplicative factor of $(1+\varepsilon)$. Therefore, the appropriately modified version of $A^{\text{sim-prof}}_{\varepsilon}$ that discards nodes with profiles equivalent (in the new sense) to already visited nodes' profiles returns a $(1+\varepsilon)^2$-approximate solution to the problem.

\begin{figure}[ht]
    \centering
    \includegraphics[width=\linewidth]{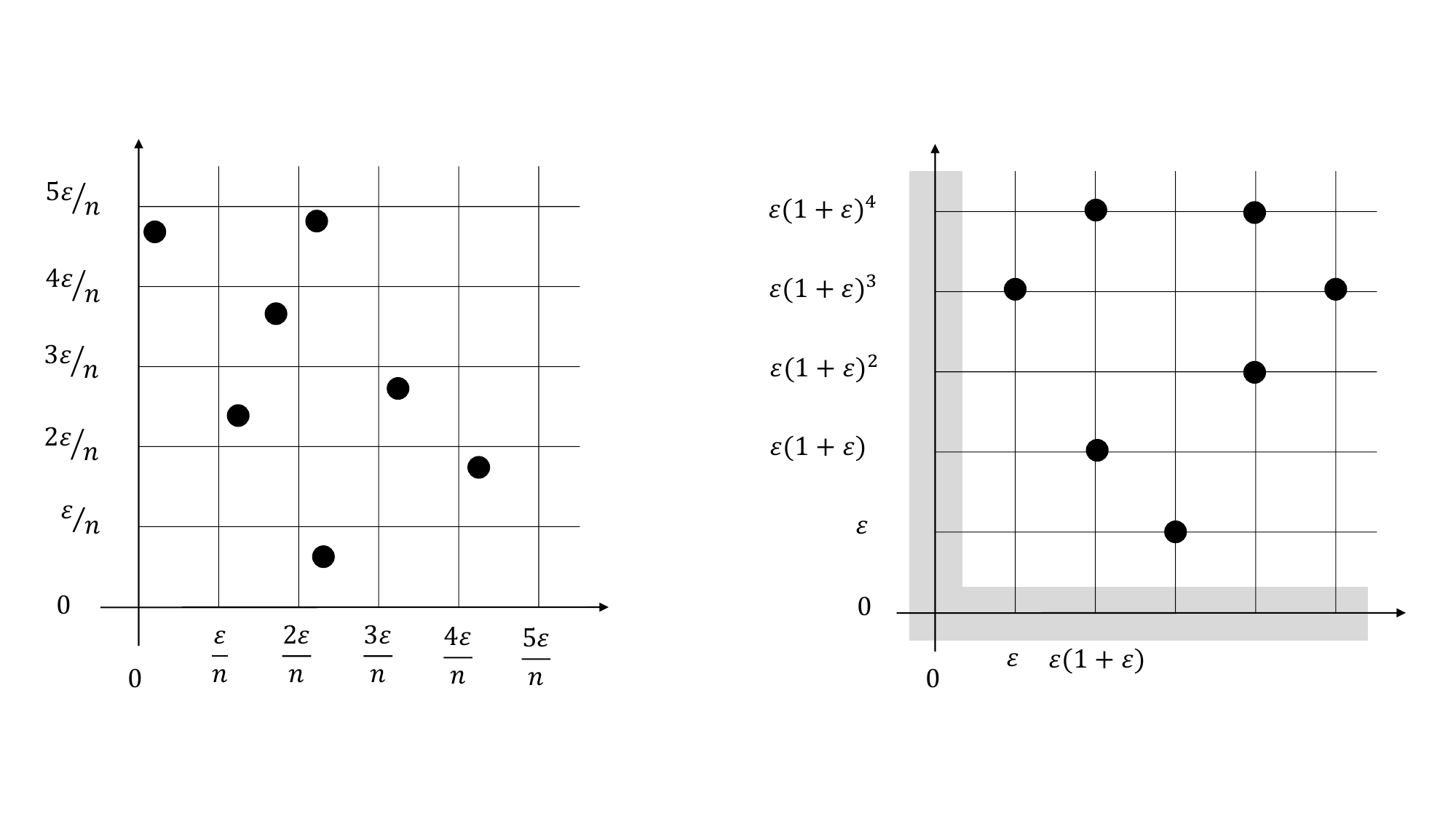}
    \caption{The notion of $\varepsilon$-similar (left), and $\varepsilon$-equivalent profiles (right). Each dot represents a potential profile vector with two coordinates. The first coordinate corresponds to the $x$-axis, the second to the $y$-axis. Profiles in the same cell (left) or the same grid point (right) are $\varepsilon$-similar/equivalent; one of them can be discarded without decreasing worst-case guarantees beyond $(1+\varepsilon)$ unless they are in the gray area. The number of cells/grid points only depends on $\varepsilon$.}
    \label{profiles}
\end{figure}

We define $A^{\text{eq-prof}}_{\varepsilon}$ as the algorithm that is identical to $A^{\text{sim-prof}}_{\varepsilon}$ in all but one aspect: it discards nodes based on $\varepsilon$-equivalency instead of $\varepsilon$-similarity.

\ptasm

\begin{proof}{\textbf{Proof of Theorem \ref{ptasm}.}}
    Let $d$ denote the number of big jobs. Since $T_{opt}\le 2$ by \eqref{const_norm}, it follows that $d\le \frac{2m}{\varepsilon}$. Suppose that at some point in the algorithm, a small job is being used for branching. Let $LLB'$ and $LUB'$ denote the global lower and upper bounds at that specific point. Rounding up a fractional solution in which each fractional job is small cannot increase the makespan by more than $\varepsilon$, so $LUB' \le LLB' + \varepsilon \le (1+\varepsilon)\cdot LLB'$, where we used that $LLB' \ge LB(root)=1$ by \eqref{const_norm}. Hence, the algorithm stops no later than reaching depth $d$. It follows from the previous analysis that each level of the branching tree contains at most as many nodes as the maximal number of pairwise non-$\varepsilon$-equivalent profiles, which is bounded by $m^{f(\varepsilon)}$. Therefore, the size of the branching tree visited by $A^{\text{eq-prof}}_{\varepsilon}$ is at most $\frac{2m^{f(\varepsilon)+1}}{\varepsilon}$, and we just need to guarantee the quality of the solution returned by the algorithm. We will show that the final solution is $(1+\varepsilon)^2$-approximate.

    Let $F^{\text{eq-prof}}_{\varepsilon}$ denote the branching tree at termination of $A^{\text{eq-prof}}_{\varepsilon}$; let $F^{\text{eq-prof}}_{\le d}$ denote the highest $d$ levels of the final branching tree given by the variant of $A^{\text{eq-prof}}_{\varepsilon}$ that does not have the early stopping criterion $\frac{LUB}{LLB}\le 1+\varepsilon$; and let $F^{\text{full}}_{\le d}$ be the tree which we obtain by running, up to depth $d$, $A^{\text{eq-prof}}_{\varepsilon=0}$ on the same input (in other words, $F^{\text{full}}_{\le d}$ is the resulting B\&B tree of an algorithm that does not consider profiles' equivalency or early stopping criteria, and finds the best partial schedule of the $d$ big jobs). It clearly holds that $F^{\text{eq-prof}}_{\varepsilon} \subseteq F^{\text{eq-prof}}_{\le d} \subseteq F^{\text{full}}_{\le d}$. Let $F^{\text{eq-prof}}_i$ and $F^{\text{full}}_i$ denote the $i$-th levels of the corresponding trees $F^{\text{eq-prof}}_{\le d}$ and $F^{\text{full}}_{\le d}$. The backbone of the proof are the corresponding versions of Lemmas \ref{i_eps} and \ref{similar_lb} for $\varepsilon$-equivalency. With a little abuse of definitions, we write $u \simeq_{\varepsilon} v$ if $\Pi(u) \simeq_{\varepsilon} \Pi(v)$. Recall that in each node in $F^{\text{full}}_i$, the $i$-th largest job is being fixed, hence the definition of profiles is only valid for nodes until depth $d$.  

    \begin{lmm}\label{repet}
        Let $0 < i \le d$ be an arbitrary integer, and let $v\in F^{\text{full}}_i$. Then there exists $u \in F^{\text{eq-prof}}_i$ such that $u \simeq_{\varepsilon} v$.
    \end{lmm}

    \begin{proof}{\textbf{Proof.}}
        We prove by induction on $i$. The case $i=1$ trivially holds, as the root node is included in both $F^{\text{full}}$ and $F^{\text{eq-prof}}$. Let $v^{parent}$ be the parent of $v$, and assume that $v$ is the $k$-th child of $v^{parent}$ from the left (equivalently, the algorithm fixed job $i$ on machine $k$ in $v^{parent}$). By induction, there exists a node $w \in F^{\text{eq-prof}}_{i-1}$ such that $v^{parent} \simeq_{\varepsilon} w$. Let $l$ be a machine whose completion time in $round(w)$ equals to the completion time of machine $k$ in $round(v^{parent})$, as guaranteed by the definition of $\varepsilon$-equivalency. Let $u$ be the $l$-th child of $w$ in $F^{\text{full}}_i$ (the one we get by fixing job $i$ on machine $l$ at node $w$). By construction, $u \simeq_{\varepsilon} v$. If $u\in F^{\text{eq-prof}}_i$, then we are ready. If not, then there exists $u' \in F^{\text{eq-prof}}_i$ such that $u \simeq_{\varepsilon} u'$, so $u' \simeq_{\varepsilon} v$ thanks to $\varepsilon$-equivalency being an equivalence-relation. \hfill \Halmos
    \end{proof}

    \begin{lmm}\label{equal_lb}
        Let $u, v \in F^{\text{full}}_i$ such that $u \simeq_{\varepsilon} v$. Then 
        \[
        \frac{1}{1+\varepsilon} \le \frac{UB(u)}{UB(v)}, \frac{LB(u)}{LB(v)}\le 1+\varepsilon.
        \]
    \end{lmm}

    \begin{proof}{\textbf{Proof.}}
        We can assume that the machines are reordered such that the completion times in $round(u)$ and $round(v)$ are pairwise equal. Let $x^*(u)$ and $x'(u)$ be the optimal fractional schedule corresponding to the sub-problem of node $u$, and its rounded-up integer complete schedule, respectively. Similarly define $x^*(v)$ and $x'(v)$. Let us remind that for an arbitrary fractional/integer schedule $S$, $C_S(i)$ denotes the completion time of machine $i$ in $S$.

        Let $i$ be an arbitrary machine. We know that for any schedule $S$ without small jobs, it holds that $C_{round(S)}(i) \le C_S(i) \le (1+\varepsilon)\cdot C_{round(S)}(i)$; and since $C_{round(u)}(i) = C_{round(v)}(i)$, it holds that 
        \[
        \frac{1}{1+\varepsilon} \le \frac{C_u(i)}{C_v(i)} \le 1+\varepsilon.
        \]
        Suppose that an optimal fractional extension of $u$ increases the completion time of machine $i$ by some $X_i$: $C_{x^*(u)}(i) = C_u(i) + X_i$. Since $u$ and $v$ consist of the same jobs, extending $v$ with the same fractional job-machine pairings yields a fractional profile whose completion time on machine $i$ have increased by $X_i$. Let us assume that the longest machine in this extension of $v$ is the first one, and so the makespan of this fractional schedule is $C_v(1) + X_1$. Therefore, the optimal fractional extension of $v$ has a makespan of at most $C_v(1) + X_1$, so
        \[
        LB(v) \le C_v(1) + X_1 \le (1+\varepsilon)\cdot C_u (1) + X_1 \le 
        \]
        \[
        \le (1+\varepsilon) \cdot (C_u(1) + X_1) = (1+\varepsilon) \cdot C_{x^*(u)}(1) \le (1+\varepsilon)\cdot LB(u).
        \]
        The other $3$ inequalities can be derived similarly. \hfill \Halmos
    \end{proof}

    With the help of these two lemmas, we can finalize the proof of the theorem. Let $LLB(F)$ and $LUB(F)$ denote the global lower and upper bounds of an arbitrary branching tree $F$, respectively. The relationship between $F^{\text{eq-prof}}_{\varepsilon}$ and $F^{\text{eq-prof}}_{\le d}$ guarantees that $LLB(F^{\text{eq-prof}}_{\varepsilon}) \le LLB(F^{\text{eq-prof}}_{\le d})$. Lemmas \ref{repet} and \ref{equal_lb} guarantee that $LLB(F^{\text{eq-prof}}_{\le d})\le (1+\varepsilon)\cdot LLB(F^{\text{full}}_{\le d})$, since if $v^{opt}$ denotes the leaf in $F^{\text{full}}_{\le d}$ with the lowest lower bound, then the lemmas guarantee the existence of a leaf node $u$ in $F^{\text{eq-prof}}_{\le d}$ such that $LB(u)\le (1+\varepsilon) \cdot LB(v^{opt})$. Combining these with the stopping criterion of $A^{\text{eq-prof}}_{\varepsilon
    }$, we have that
    \[
    LUB(F^{\text{eq-prof}}_{\varepsilon}) \le (1+\varepsilon)\cdot LLB(F^{\text{eq-prof}}_{\varepsilon}) \le (1+\varepsilon)\cdot LLB(F^{\text{eq-prof}}_{\le d})\le 
    \]
    \[
    \le (1+\varepsilon)^2 \cdot LLB(F^{\text{full}}_{\le d}).
    \]
    As $LLB(F^{\text{full}}_{\le d})$ is a true lower bound on $T_{opt}$ (since the leaves of $F^{\text{full}}_{\le d}$ provide a complete partition of feasible solutions), the proof is concluded. \hfill \Halmos
\end{proof}

\section{Computational Experiments}\label{sec:comp_exp}
In this section, we aim to assess the performance of our proposed algorithm on some randomly generated instances. Specifically, we compare our proposed strategies, which gave us theoretical guarantees, with other commonly used. The goal is to assess whether our theoretical guarantees are also observable in practice. We also provide a detailed runtime analysis. For the instances under study, a carefully optimized B\&B implementation, such as SCIP \cite{scip}, outperforms our naive implementation. However, we chose to reimplement everything from scratch, focusing on simplicity rather than efficiency.

\textbf{Experimental Setting.} All experiments were conducted on a Linux computer equipped with Intel Xeon E5-2650 v3 CPUs, each running at 2.3 GHz, and 64 GB of RAM. Our main code was implemented in Python 3.10.14, and all optimization routines were carried out using SCIP \cite{scip}. The code is provided as supplementary material.

\subsection{Multiple knapsack problem}

To test our algorithm, we generate 30 random instances for each pair $(n, m) \in \{(5, 2), (10, 2), (10, 5), (50, 2), (50, 5), (50, 15), (100, 2), (100, 5), (100, 10), (100, 15)\}$. Capacities are uniformly sampled integers from the range $[c_{\min}, c_{\max}]$, where $c_{\min} = \min_{j}{w_j}$ and $c_{max} = \left\lceil \frac{\sum w_j}{n} \right\rceil - c_{\min}$. The lower bound ensures that each item fits inside at least one of the knapsacks, while the upper bound ensures that (on average) half of the items fit in the union of the knapsacks, as discussed in \cite{chvatal}.

As baselines for node selection, we test DFS and BFS alongside the HUB rule. For branching rules, we evaluate two approaches in addition to the previously introduced ``critical element'' (CE) strategy. In one strategy, we branch on the items among the \emph{fractional} ones with the largest profit-to-weight ratio (PPW). In the other strategy, as suggested by \cite{kolesar}, we branch on the item among the \emph{unfixed} ones with the largest profit-to-weight ratio (K).
We test these strategies for different values of $\alpha$ and collect various metrics, including the number of nodes explored, the gap to the optimum, the maximum depth reached, the number of nodes after finding the optimum, and the number of left turns.
Here, we report partial results, while a more extensive set of experiments is available in the interactive notebook. 

Figure \ref{fig:mk_nodes} shows the number of nodes explored to get an $\alpha = 0.97$ approximation. 
Since our implementation is a proof of concept and not fully optimized, we encountered memory issues. To address this, we imposed a threshold of $10^4$ nodes explored, beyond which we return the best solution found so far.
We observe that, in terms of the number of nodes explored, the Highest Upper Bound (HUB) strategy consistently outperforms the others. 
This is particularly evident in the ``hard'' instances $(100, 10), (100, 15), (50, 15)$, where all successful methods in at least one instance involve the HUB strategy.
Our proposed strategy (yellow box) frequently achieves the best overall performance.
Interestingly, in several cases, branching using the PPW rule yields better results compared to branching based on the Critical Element (CE) criterion.

In our analysis, we also record the optimality gap of the returned solution, defined as 
\[\frac{|z - z^*|}{\max(z,z^*)} \]
where $z$ is the solution as returned by our algorithms and $z^*$ is the optimal solution we computed using state-of-the-art Google OR-Tools \cite{ortools} with SCIP \cite{scip} as a linear solver. 

Figure \ref{fig:mk_gap} presents this information, clearly showing that HUB is often a winning strategy in terms of producing high-quality solutions. In this case, we do not observe any significant difference between CE and PPW.

Lastly, Figure~\ref{fig:mk_our} illustrates the behavior of our algorithm for different numbers of knapsacks as the number of items increases.

\subsection{Unrelated machine scheduling problem}
In this case, we generate 30 random instances for each pair $(n, m) \in \{(5, 2), (10, 2), (10, 5), (50, 2), (50, 5), (50, 10), (100, 2)\}$. Job lengths are uniformly sampled integers from the range $[1, 100]$. 
Note that, in this case, the analysis of the $(50, 5)$ instance could not be completed within our 48-hour time frame. Hence, we report only the average over the instances that were successfully solved.  
We attempt to understand why this occurred, given that the Multi-Knapsack framework initially seemed more tractable.  
In this case, the binary search involves repeatedly solving LPs, significantly increasing computational overhead. We have 12 different B\&B-like algorithms to evaluate, whereas in the Multi-Knapsack setting, there were only 9. Lastly, unlike Multi-Knapsack, there is no pruning by infeasibility, making it harder to discard unpromising nodes quickly.

As baselines for node selection, we test DFS and BFS along with the proposed rule LLB.
For the lower bound, we chose both the proposed BS scheme and the Linear Relaxation (LR) of Integer Linear Programming, minimizing the makespan that is commonly used in unrelated parallel machine scheduling.  
%
As the branching rule, we test only the one we propose: branching on the variable with the largest minimum processing time across machines (MMP).  
Both BS and LR return a solution that may contain fractional components, which we need to round to obtain an upper bound on the optimal solution.  
We compare two different rounding strategies (i) The one we prove leads to a PTAS, which assigns \emph{A}ll fractional jobs to the machine where their processing time is \emph{S}hortest (AS); (ii) An alternative approach based on Best Matching (BM) of the at most $m$ fractional jobs, where we find a matching that minimizes the total makespan. 
Even in this case, we observe a similar trend: LLB results in fewer nodes explored. However, on average, BM appears to be a better rounding strategy compared to AS.  
Regarding the optimality gap, interestingly, BFS slightly outperforms LLB in some instances (e.g., $n = 100, m = 5$). This is not surprising, as our theoretical results (Proposition \ref{bfs}) suggest that BFS also guarantees a PTAS.
Overall, we observe a significant difference in the order of magnitude of the average gap between the two problems. In the first case, a gap of 0 is rarely achieved, whereas in the Unrelated Machine Scheduling problem, the algorithm reaches the optimal solution for $n \in \{5, 10\}$ and for the instance $(100, 2)$, regardless of the choice of branching, bounding, and selection strategy.

It is worth noting that BS provides a stronger lower bound than the standard dual bound given by LR, not only in theory but also in practice. For instance, Figure~\ref{fig:ujs_nodes} shows that, regardless of the node selection strategy, the best-performing approach is always associated with a methodology that uses BS.

Lastly, Figure~\ref{fig:ujs_our} illustrates the behaviour of our algorithm for different numbers of machines as the number of jobs increases.

\begin{figure}
\FIGURE
{
\subcaptionbox{Multi knapsack, $\alpha = 0.97$.\label{fig:mk_nodes}}
{\includegraphics[width=0.5\textwidth]{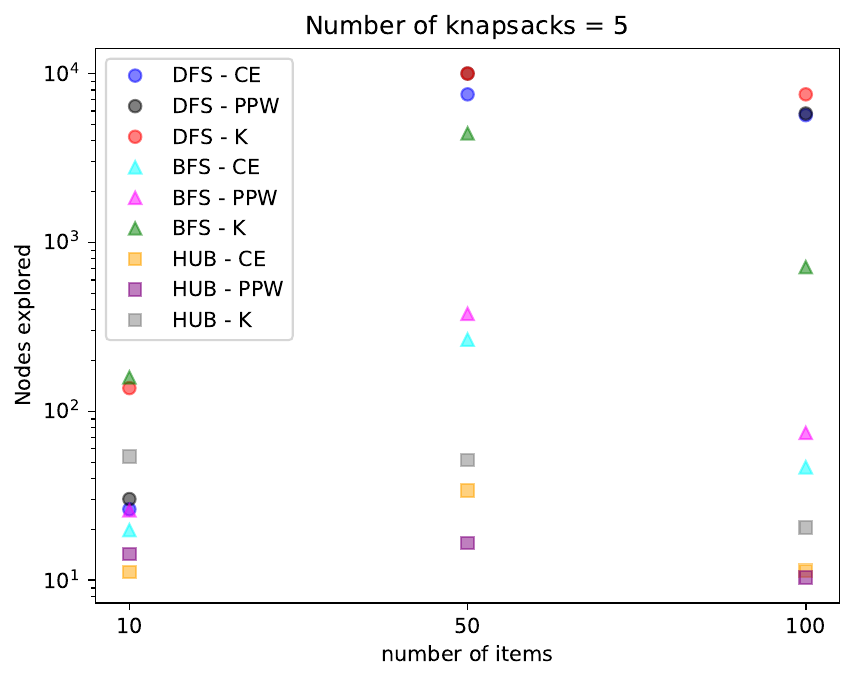}}
\hfill\subcaptionbox{Unrelated Machine Scheduling, $\varepsilon = 0.01$.\label{fig:ujs_nodes}}
{\includegraphics[width=0.5\textwidth]{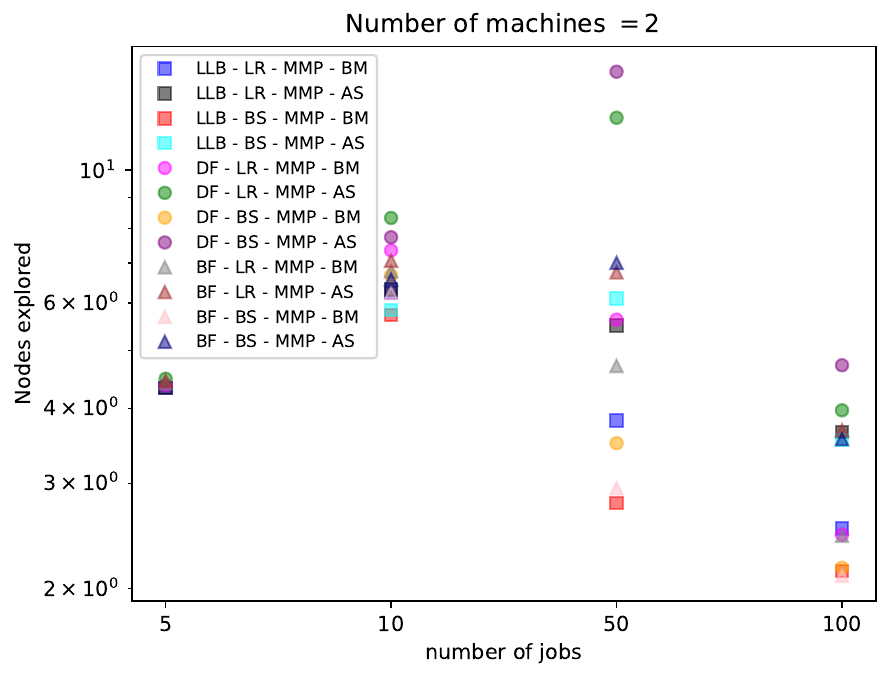}}
}
{
Performance of different strategies in the branch-and-bound method for the Multi-Knapsack (with just $5$ knapsacks) and Unrelated Machine Scheduling problems (with just 2 machines). The geometric mean of the number of nodes explored before termination (or reaching the stopping condition) is reported on a logarithmic scale.
\label{fig:nodes}}
{
}
\end{figure}

\begin{figure}
\FIGURE
{
\subcaptionbox{Multi knapsack, $\alpha = 0.97$.\label{fig:mk_gap}}
{\includegraphics[width=0.5\textwidth]{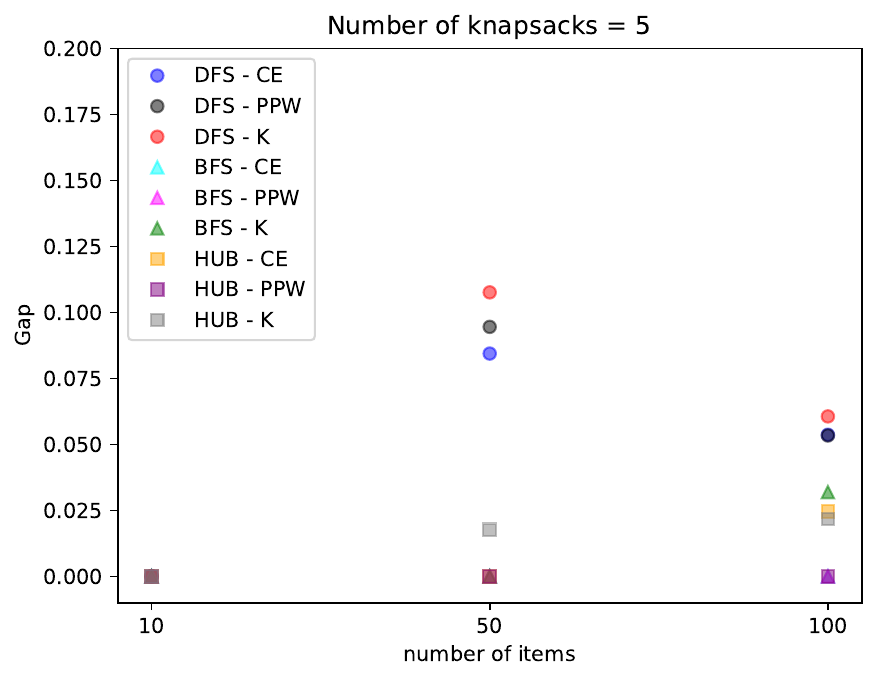}}
\hfill\subcaptionbox{Unrelated Machine Scheduling, $\varepsilon = 0.01$.\label{fig:ujs_gap}}
{\includegraphics[width=0.5\textwidth]{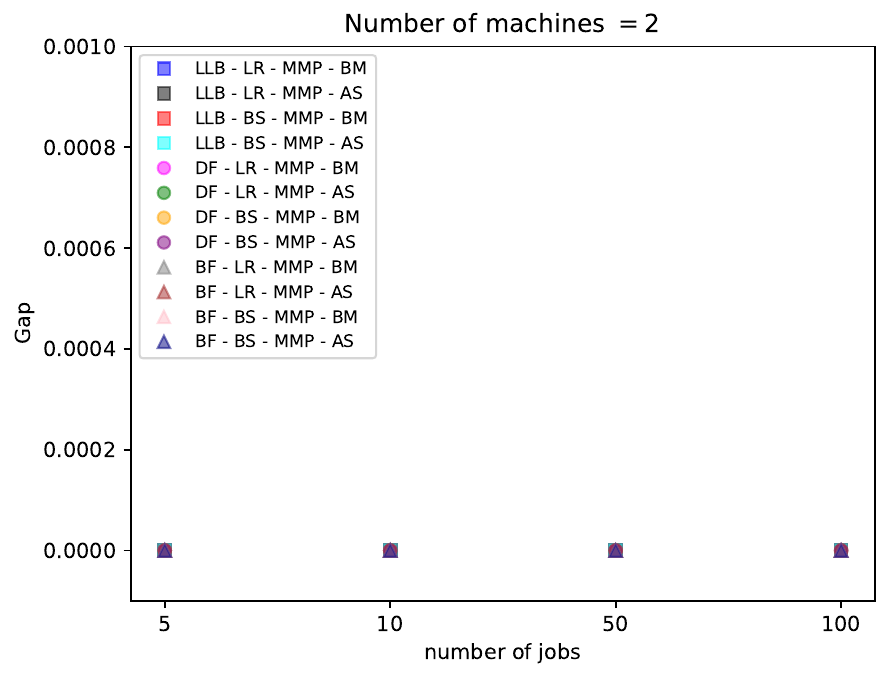}}
}
{
Performance of different strategies in the branch-and-bound method for the Multi-Knapsack (with just $5$ knapsacks) and Unrelated Machine Scheduling problems (with just 2 machines). The geometric mean of the optimality gap reported as obtained before termination (either upon reaching the stopping condition or exiting). The $y$-axis has been limited to highlight the most relevant portion of the plot.
\label{fig:nodes_2}}
{
}
\end{figure}

\begin{figure}
\FIGURE
{
\subcaptionbox{Optimality gap, $\alpha = 0.97$.\label{fig:mk_gap_oa}}
{\includegraphics[width=0.5\textwidth]{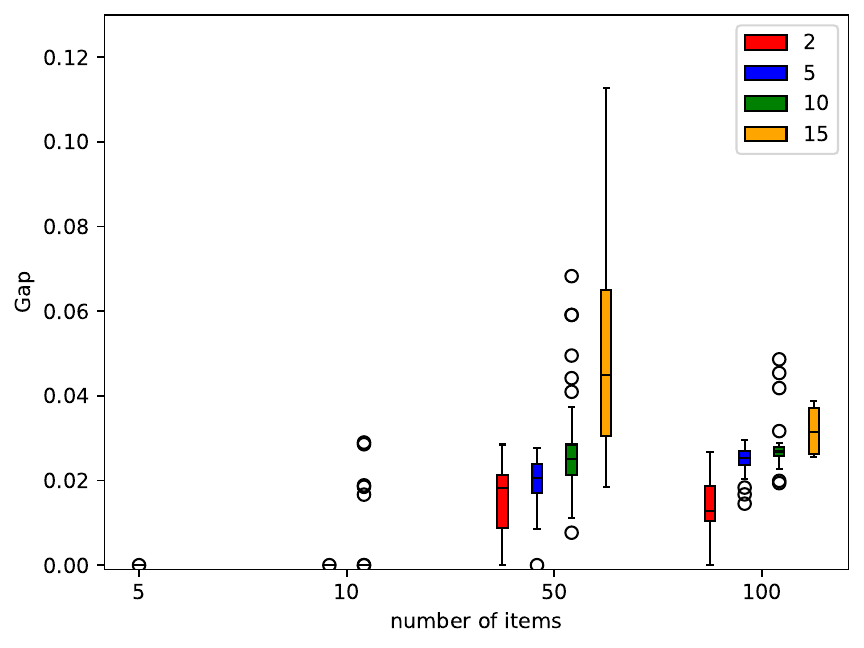}}
\hfill\subcaptionbox{Nodes explored, $\alpha = 0.97$. \label{fig:mk_nodes_oa}}
{\includegraphics[width=0.5\textwidth]{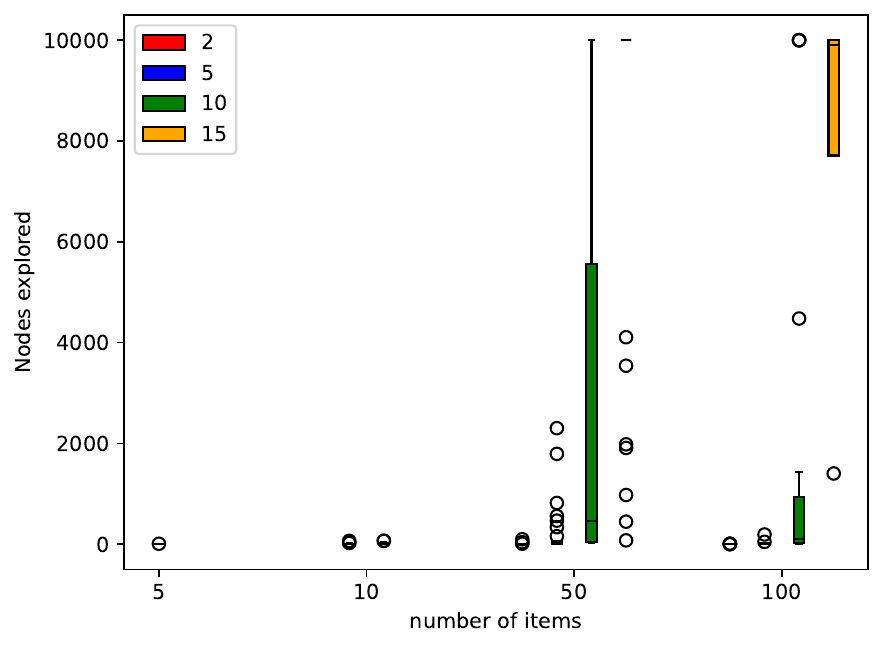}}
}
{
Performance of our branch-and-bound method for the Multi-Knapsack problem, comparing instances with different numbers of knapsacks ($\{2,5,10,15\}$). Missing data indicates that the corresponding item–knapsack configuration is not present in our dataset. We report the distributions of two parameters: the number of nodes explored before termination (or reaching the stopping criterion) and the optimality gap at termination.
\label{fig:mk_our}}
{
}
\end{figure}

\begin{figure}
\FIGURE
{
\subcaptionbox{Optimality gap, $\varepsilon = 0.01$.\label{fig:ujs_gap_oa}}
{\includegraphics[width=0.5\textwidth]{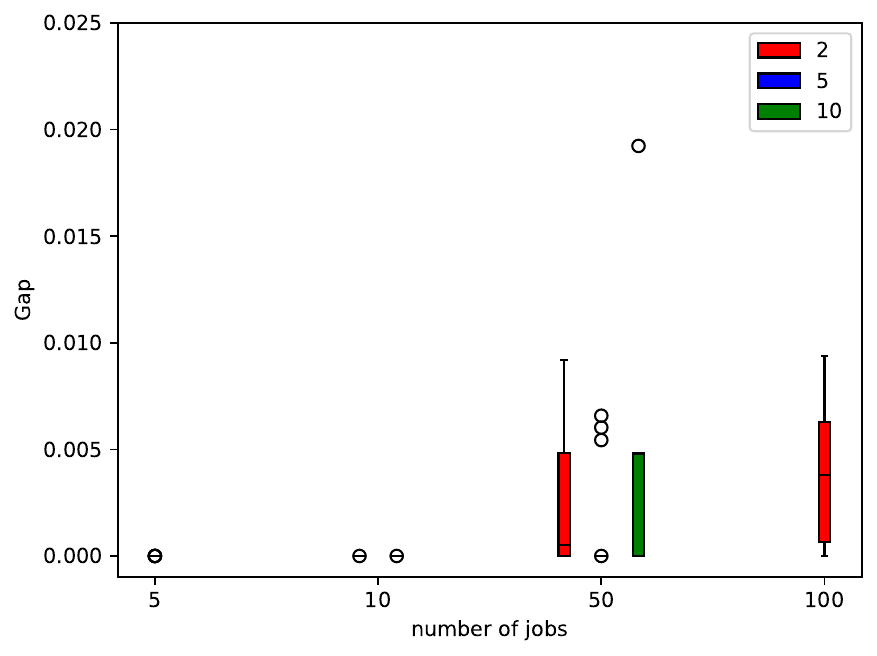}}
\hfill\subcaptionbox{Nodes Explores, $\varepsilon = 0.01$. \label{fig:ujs_nodes_oa}}
{\includegraphics[width=0.5\textwidth]{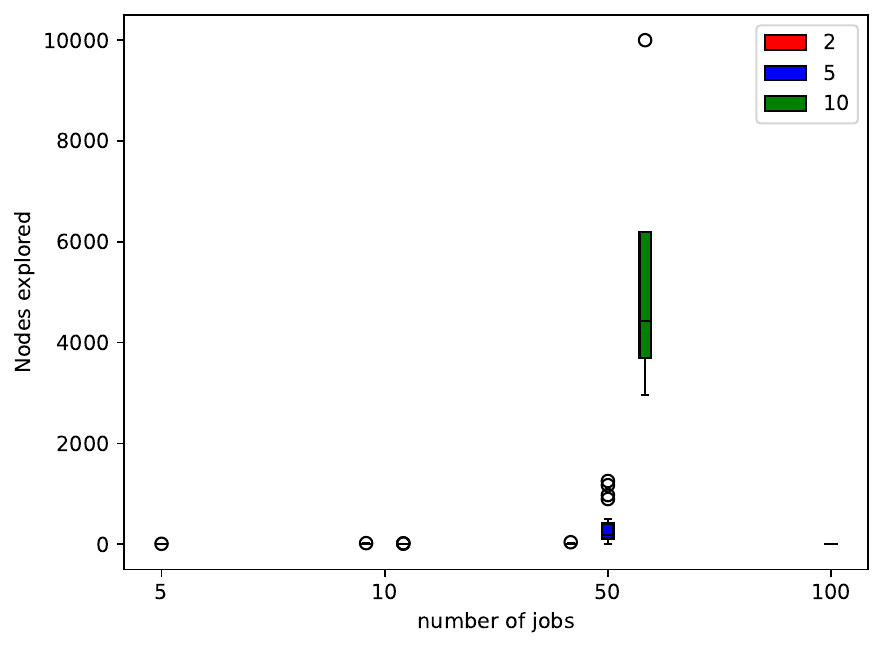}}
}
{
Performance of our branch-and-bound method for the Unrelated Machine Scheduling problem, comparing instances with different numbers of machines ($\{2,5,10\}$). Missing data indicates that the running time for the corresponding configuration is prohibitively large. We report the distribution of two parameters: the number of nodes explored before termination (or reaching the stopping criterion) and the optimality gap at termination.
\label{fig:ujs_our}}
{
}
\end{figure}

\subsection{Analysis of the runtime of the proposed 
algorithm}\label{app:comp_exp}

\begin{figure}
    \centering
    \includegraphics[width=0.49\linewidth]{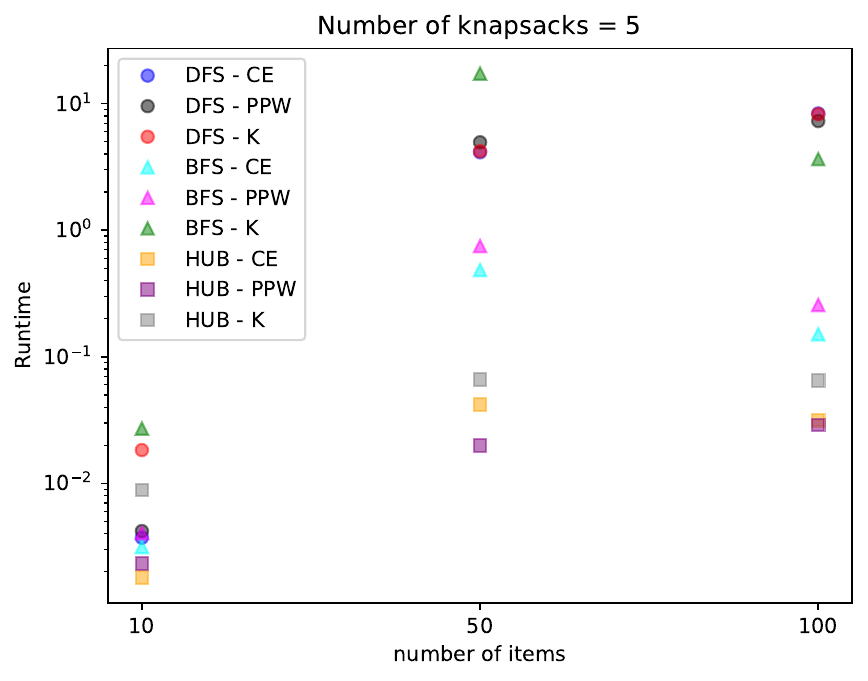}
    \includegraphics[width=0.49\linewidth]{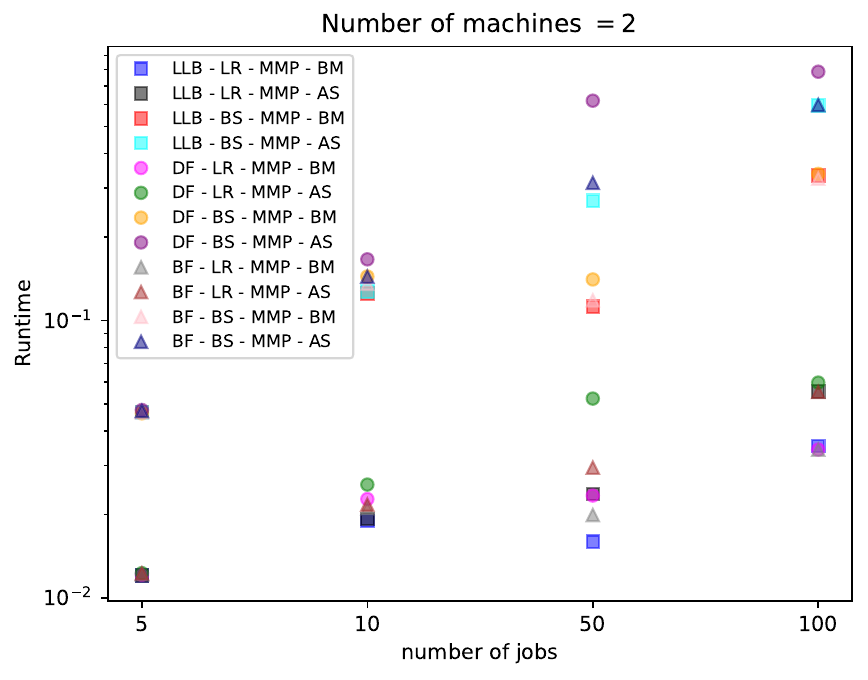}
    \caption{Performance of different strategies in the branch-and-bound method for the Multi-Knapsack (with just $5$ knapsacks) and Unrelated Machine Scheduling problems (with just 2 machines). The geometric mean of the runtime (in seconds) is reported on a logarithmic scale.}
    \label{fig:runtime}
\end{figure}

In Figure \ref{fig:runtime}, we report the runtime of our proposed algorithm for the Multi-Knapsack problem (left) and the Unrelated Machine Job Scheduling problem (right).  

\subsubsection{Multiple Knapsack}  
First, we observe that the runtime is quite long, even for relatively small instances in both cases.  
As previously mentioned, our implementation is not highly optimized.  
To put this into perspective, solving the benchmark instance $(100, 15)$ in the Multi-Knapsack framework typically requires $7.82 \pm 4.40$ seconds with the B\&B implemented in SCIP, which is significantly shorter than our results (in this case, for fairness, we disable presolve, cutting plane, heuristics, restarts, and propagation, and set a branching nodes limit equal to $10^4$).

However, an interesting pattern emerges from the analysis of our strategies.  
For the Multi-Knapsack problem, we observe that when the number of knapsacks is $\geq 10$, DFS outperforms HUB in terms of speed.  
We suspect this is because, in the Multi-Knapsack setting, the DFS strategy consistently reaches the node exploration limit, set to $10^4$.  
As a result, it likely prunes many nodes quickly due to infeasibility or bounding, avoiding complex computations on the explored nodes. Hence, we arrive at the node limit faster. 

\subsubsection{Unrelated machine scheduling problem}
As in the previous section, we compare our approach with SCIP for the case $(50, 10)$, where we observe a runtime of $10.051\pm 1.63$, which is again shorter than the observed runtime of our algorithm.  

In this case, we find that the LLB strategy results in faster solutions.  
This outcome is expected: since all nodes are explored and pruning by infeasibility is not possible (as every subproblem remains feasible), the time required per node remains roughly the same.  

Since the LLB strategy achieves better solutions with fewer node explorations, we can conclude that our theoretical expectations align with the experimental results.

\section{Concluding Remarks}\label{sec:conc}

Let us collect general observations about our methods. Most importantly, we note that the best-first rule can be replaced by BFS in $A^{\text{unrel}}_{\varepsilon}$ without losing the theoretical worst-case guarantee (while keeping the other parameters fixed). The key element of our proof was to limit the depth of the final tree of the algorithm ($F$) at $\lfloor\frac{m^2}{\varepsilon}\rfloor$. Let $F$ be the branching tree of $A^{\text{unrel}}_{\varepsilon}$, and let $F'$ be the B\&B tree of the alternative with BFS limited to depth $\lfloor\frac{m^2}{\varepsilon}\rfloor$. Since $F \subseteq F'$, the lowest lower bound in $F'$ (denoted by $LLB'$) is greater or equal than the lowest lower bound in $F$ (denoted by $LLB$), and the best integer solution found in $F'$ (denoted by $LUB'$) is better than the best integer solution we found with $A^{\text{unrel}}_{\varepsilon}$ (denoted by $LUB$). Hence,
\[
\frac{HUB'}{LLB'} \le \frac{HUB}{LLB}\le 1+\varepsilon,
\]
and it follows that with BFS as the search strategy (denoted by $A^{\text{BFS-unrel}}_{\varepsilon}$), we terminate no later than processing $F'$. Since $|F'|\le m^{\lfloor\frac{m^2}{\varepsilon}\rfloor}$, we have that

\begin{prop}\label{bfs}
    For every fixed $\varepsilon > 0$, the algorithm $A^{\text{BFS-unrel}}_{\varepsilon}$ returns a $(1+\varepsilon)$-approximate solution to the unrelated machine scheduling problem, after processing at most $m^{\lfloor\frac{m^2}{\varepsilon}\rfloor}$-many nodes in the branching tree.
\end{prop}

However, according to our experiments, best-first seems empirically better in terms both of optimality gap and number of nodes explored. For $A^{\text{knap}}_{\alpha}$, we do not have similar guarantees as there we only have bound on the number of left-turns in the branching tree and the depth of $F$ can potentially be as large as $n$. For the algorithms $A^{\text{sim-prof}}_{\varepsilon}$ and $A^{\text{eq-prof}}_{\varepsilon}$, the bound on the number of visited nodes was independent of the tree traversal strategy, since our proofs only rely on the limited number of different nodes at each level. Therefore, any alternative strategy can be used to replace the best-first one with the same worst-case guarantee.

Next, we note that for the machine scheduling problem, the results and algorithms can be modified to work with the ``standard'' LP-formulation lower bound instead of the binary search one, but the lower bound itself is trivially worse. Last, we mention that for the special case of identical parallel machines, $A_{\varepsilon}^{\text{unrel}}$ can be improved with a slight change in the rounding method. This version visits $m^{\lfloor \frac{m}{\varepsilon}\rfloor}$ nodes in the worst case, by keeping the exact same argument. Furthermore, there are fast and intuitive heuristics for finding vertices of the polyhedra, thus the time spent in individual nodes can also be reduced.

We conclude by collecting the most essential properties of the knapsack and machine scheduling problems that were exploited during the investigation, intending to set the ground for generalizing the results to a larger class of problems. 

Perhaps the most paramount property the two problems have in common is the notion of \emph{self-similarity}. To repeatedly apply the same argument for each node of a path in the branching tree, we needed the sub-problems encoded by these nodes to fall into the same category as the original problem. In other words, fixing one variable to $1$ or $0$ should result in a problem that is in the same class as the original one. This property was by default true for the knapsack problem: setting $x_{j, i} = 1$ in $MK_m(\bm{C}, \bm{w}, \bm{p})$ yields the sub-problem $MK_m(\bm{C'}, \bm{w}|_{[n]-j}, \bm{p}|_{[n]-j})$ with $\bm{C'} = (C_1, \ldots, C_{i-1}, C_i - w_j, C_{i+1}, \ldots, C_m)$, while the rightmost branch corresponds to $MK_m(\bm{C}, \bm{w}|_{[n]-j}, \bm{p}|_{[n]-j})$. For the machine scheduling problem, on the other hand, the default description was not sufficient. If we fix a binary variable to $1$ in the standard linear programming formulation, the resulting LP will not correspond to a machine scheduling problem of the same type. However, introducing overheads ensures the desired property, since now setting a variable to $1$ corresponds to increasing the appropriate machine's overhead by the processing time of the fixed job.

Strongly related to this property, we relied on the monotonicity of subproblems: for maximization problems, the local upper bound of a node is greater than the local upper bound of any of its children (for minimization problems, a similar property holds). However, this is a direct consequence of using the same objective function on subsequently smaller sets.

We also exploited that there was a quantifiable relationship between a node's lower/upper bounds and the job/item that was fixed at the node. For the knapsack problem, this relationship is guaranteed by Lemma \ref{knapsack_gap_and_critical_element}, whereas for the machine scheduling problem, the inequality
\[
\frac{p'_j}{LB(v)} \le \frac{m}{k}
\]
provided the connection.

Finally, the least demanding requirement we need to pose is that of approximability. When rounding a fractional solution of a sub-problem at a node, an $(m+1)$-approximation algorithm was used for both the knapsack problem and the machine scheduling problem. However, for the machine scheduling problem, the proof did not rely upon this local rounding guarantee, and hence the necessity of a constant-factor approximation rounding is unclear.

\bigskip

It is important to acknowledge the limitations and drawbacks of our approach. During the last couple of decades, several approximation schemes have been described for both the knapsack and the machine scheduling problem. In fact, both problems are known to admit a fully polynomial-time approximation scheme (FPTAS), which is far superior to the PTAS framework in which the desired proximity ratio ($\varepsilon$ or $\alpha$) appears in the exponent of the running time. Furthermore, as we see in Subsection \ref{sec:js_fptas_proof}, our arguments are not directly repeatable or extendable for some of the cases. Nevertheless, we believe that the connection between B\&B and approximation algorithms explored in the paper adds a surprising flavor to the theory of branch-and-bound algorithms, and sheds some light on their good behavior observed in practice.

For future research directions, we mention the possibility of a B\&B yielding an FPTAS for the unrelated machine scheduling problem (or even more complex scheduling paradigms such as the job shop problem), with a possibly different choice of parameters and additional rounding tricks.

\bibliographystyle{informs2014}

\bibliography{references}

\begin{APPENDICES}

\section{Proofs}\label{app:proofs}

\subsection{Proof of Lemma \ref{vertex_char}.}
    First, we will prove that if $\bm{x}$ is a vertex, then (i) and (ii) hold true. 
    By contradiction, assume that $G(\bm{x})$ is a proper pseudo forest; that is, there exists a connected component that contains a cycle. By relabeling jobs and machines, we can assume that the edges of the cycle are given by the non-zero fractional variables $x_{11}, x_{21}, x_{22}, x_{32}, \ldots, x_{k,k}, x_{1,k}$. We will show that $\bm{x}$ is not a vertex by proving that both $(\bm{x} + \bm{\varepsilon})$ and $(\bm{x} - \bm{\varepsilon})$ are feasible for a vector $\bm{\varepsilon}$ of length $n\cdot m$. Clearly, $\bm{\varepsilon}_{ij} = 0$ must hold every time $x_{ij} = 0$ or $x_{ij}=1$, so it is enough to consider variables corresponding to edges of $G(\bm{x})$. We claim that we can find an appropriate $\bm{\varepsilon}$ that is non-zero only on these edges. For this to hold, we clearly need that $\varepsilon_{11} = - \varepsilon_{1,k}, \ldots, \varepsilon_{k,k-1} = - \varepsilon_{k,k}$, so we can write the desired vector in the form $\varepsilon_{11} = \varepsilon_1, \varepsilon_{1,k} = -\varepsilon_1, \ldots, \varepsilon_{k,k-1} = \varepsilon_k, \varepsilon_{k,k} = - \varepsilon_k$.

    In order for $\bm{x} + \bm{\varepsilon}$ to have the same (or smaller) makespan as $\bm{x}$, we need the following conditions to be satisfied for each $i=1,\ldots,k$ (for $i=k$, $i+1$ is to be understood as $1$):
    \[
    (x_{i,i} + \varepsilon_i)\cdot p_{i,i} + (x_{i+1,i}- \varepsilon_{i+1})\cdot p_{i+1,1} \le x_{i,i}\cdot p_{i,i} + x_{i+1,i}\cdot p_{i+1,i},
    \]
    or equivalently,
    \[
    \varepsilon_i \cdot p_{i,i} - \varepsilon_{i+1}\cdot p_{i+1,1} \le 0,
    \]
    and
    \[
    \frac{\varepsilon_i}{\varepsilon_{i+1}} \le \frac{p_{i+1,i}}{p_{i,i}}. 
    \]
    Since we are in the uniform case, we know that $p_{j,k} = \frac{p_j}{s_k}$ with a profit $p_j$ and machine speed $s_k$, and so we need that 
    \begin{equation}\label{plus_eps}
        \frac{\varepsilon_i}{\varepsilon_{i+1}} \le \frac{p_{i+1}/s_i}{p_i/s_i} = \frac{p_{i+1}}{p_i}, \,\,\,\,\, i = 1, \ldots, k.
    \end{equation}
    Let us choose $\varepsilon_1 > 0$ arbitrarily, then we recursively define 
    \[
    \varepsilon_{i+1} := \varepsilon_i \cdot \frac{p_i}{p_{i+1}}.
    \]
    These trivially satisfy with equality each inequality from \eqref{plus_eps} apart from the $i=k$ case; but since the product of all the left-hand sides, as well as the product of all right-hand sides, is equal to $1$, the remaining inequality must be satisfied (with equality) as well.

    Apart from \eqref{plus_eps}, the values of the parameters $\varepsilon_i$ must adhere to the constraints
    \begin{equation}\label{plus_varepsilon_2}
    \begin{cases}
        x_{i,i} + \varepsilon_i \le 1, & i = 1, \ldots k, \\
        x_{i+1,i} -\varepsilon_i \ge 0, & i = 1, \ldots, k.
    \end{cases}
    \end{equation}
    But these are satisfied if the values for $\varepsilon_i$ are chosen to be small enough since $0 < x_{i,i}, x_{i+1, i} < 1$. Notice that multiplying each $\varepsilon_i$ with the same constant does not change the fractions in \eqref{plus_eps}; so by choosing an appropriately small constant, we can guarantee that both \eqref{plus_eps} and \eqref{plus_varepsilon_2} are satisfied.

    Repeating a similar reasoning, we gather that the same $\bm{\varepsilon}$ satisfies the corresponding versions of \eqref{plus_eps} for $\bm{x}-\bm{\varepsilon}$ as well, and if $\bm{0} \le \bm{x}-\bm{\varepsilon} \le \mathbf{1}$ is not satisfied, we can again multiply $\bm{\varepsilon}$ with a small enough constant to guarantee these inequalities while maintaining the feasibility of the makespan constraints. In conclusion, if $G(\bm{x})$ has a cycle, $\bm{x}$ cannot be a vertex.

    Now, suppose that $\bm{x}$ is a vertex and assume for contradiction that there are at least two machines in the same component of $G(\bm{x})$, $i_1$, and $i_2$, whose completion times are strictly smaller than $T'$. Consider a path connecting these two nodes, and observe that we can construct a vector $\bm{\varepsilon} \in \mathbb{R}_{\ge 0} ^{n\times m}$ which is zero apart from the coordinates of the path, and for which $(\bm{x}\pm\bm{\varepsilon})$ are feasible. To do so, notice that we can repeat the process we described for the case of having a cycle, with the sole difference that in \eqref{plus_eps} the first and last inequalities lack one of the variables, and \eqref{plus_varepsilon_2} has two less constraints. Therefore, the solution we derived for cycles is feasible for paths as well. This construction in fact can be seen as a special case of having a cycle, by splitting a ``dummy job'' between machines $i_1$ and $i_2$ and blowing up both completion times to $T'$

    Conversely, suppose that $G(\bm{x})$ does not contain a cycle, and each connected component has at most one machine with completion time strictly smaller than $T'$. Suppose for contradiction that there is a vector $\bm{\varepsilon}$ such that $(\bm{x} + \bm{\varepsilon})$ and $(\bm{x} - \bm{\varepsilon})$ are both feasible. The coordinates of $\bm{\varepsilon}$ which are different from $0$ and $1$ must correspond to edges in $G(\bm{x})$; and as it does not contain a cycle, the subgraph spanned by these edges must be a forest. Consider an arbitrary connected component of this forest having at least two nodes; this must have at least two nodes of degree $1$. If any of these nodes correspond to a job $j$, the constraints 
    \[
    \sum\limits_{i=1}^n x_{j,i} = 1
    \]
    would be violated by both $\bm{x}\pm \bm{\varepsilon}$. Therefore, we have two machines, $i_1$ and $i_2$, that have degree one and belong to the same connected component in the subgraph spanned by the fractional coordinates of $\bm{\varepsilon}$. Furthermore, since the connected components of this subgraph must be part of some connected components of $G(\bm{x})$, by assumption we have that at least one of $C_{\bm{x}}'(i_1) = T'$ or $C_{\bm{x}}'(i_2) = T'$ holds; let us assume that it is the first one. But then $C_{\bm{x}\pm\bm{\varepsilon}}'(i_1) > T'$ would hold for exactly one of $\bm{x}\pm\bm{\varepsilon}$, and so this vector would not be feasible. \hfill \Halmos

\subsection{Proof of Lemma \ref{longest_frac}.}
    If job $n$ is fractionally assigned in $\bm{x}^*$, we can choose $\hat{\bm{x}} = \bm{x}^*$. Otherwise, assume that job $n$ is integrally assigned to a machine; by relabeling machines, we can further assume that it is integrally assigned to machine $1$ and hence $x^*_{n,1} = 1$. Since $\bm{x^*}$ is fractional, there must exist a job $j$ which is assigned fractionally. We distinguish between two cases. If there is a job that has a fractional part assigned to machine $1$, then let $j$ be this job. Apart from $x^*_{j,1}$, there must be another nonzero coordinate for job $j$; we can assume that it is $x^*_{j,2}$. If there is no job having a fractional coordinate at machine $1$, let $j$ be an arbitrary fractional job and assume that $x^*_{j,2} \ne 0$.

    Let $\varepsilon_1$ and $\varepsilon_2$ be parameters whose value we fix later, and consider the following vector $\bm{\hat{x}}$:
    \[
    \hat{x}_{k,i} = \begin{cases}
        x^*_{n,1} - \varepsilon_1 = 1-\varepsilon_1, & \text{ if } (k,i) = (n,1), \\
        x^*_{n,2} + \varepsilon_1 = \varepsilon_1, & \text{ if } (k,i) = (n,2), \\
        x^*_{j,1} + \varepsilon_2, & \text{ if } (k,i) = (j,1), \\
        x^*_{j,2} - \varepsilon_2, & \text{ if } (k,i) = (j,2), \\
        x^*_{k,i}, & \text{ else}.
    \end{cases}
    \]
    Let us choose the values of $\varepsilon_1$ and $\varepsilon_2$ such that the completion times of machines $1$ and $2$ are the same in $\bm{x}^*$ and $\bm{\hat{x}}$. For this to hold, we need that
    \[
    -\varepsilon_1 \cdot \frac{p_n}{s_1} + \varepsilon_2 \cdot \frac{p_j}{s_1} = 0,
    \]
    and 
    \[
    +\varepsilon_1 \cdot \frac{p_n}{s_2} - \varepsilon_2 \cdot \frac{p_j}{s_2} = 0.
    \]
    These equalities are satisfied by any $\varepsilon_1, \varepsilon_2$ for which 
    \[
    \frac{\varepsilon_1}{\varepsilon_2} = \frac{p_j}{p_n}
    \]
    holds. Let us choose $\varepsilon_2 = \min\{1-x^*_{j,1}, x^*_{j,2}\} = x^*_{j,2} < 1$, and $\varepsilon_1 = \varepsilon_2 \cdot\frac{p_j}{p_n}$. With this choice, it also holds that $\hat{x}_{j,2} = 0$, $\hat{x}_{j,1} \le 1$; and $\varepsilon_1 = \varepsilon_2 \cdot \frac{p_j}{p_n} \le \varepsilon_2 < 1$ implies that $\hat{x}_{n,1} > 0$, $\hat{x}_{n,2} < 1$ is also true.

    In other words, the modification splits job $n$ fractionally between machines $1$ and $2$, while relocating the fractional part of job $j$ from machine $2$ to machine $1$. The resulting $\bm{\hat{x}}$ is feasible for the corresponding polyhedra; and by Lemma \ref{vertex_char}, it is enough to guarantee that $G(\bm{\hat{x}})$ is a forest for concluding that $\bm{\hat{x}}$ is a vertex, since the completion times of machines are left unchanged.

    Let $u_1, u_2, w_n, w_j$ be the nodes of $G(\bm{\hat{x}})$ corresponding to machine $1$, machine $2$, job $n$ and job $j$, respectively. We have that $w_n u_1, w_n u_2, w_j u_1 \in E(G(\bm{\hat{x}}))$ and $w_j u_2 \not \in E(G(\bm{\hat{x}}))$. We also know that $\deg _{G(\bm{\hat{x}})} (w_n) = 2$ and $w_n \not \in V(G(\bm{x}^*))$.

    Consider the first case, when job $j$ has a fractional part on machine $1$ in $\bm{x}^*$. Suppose for contradiction that there is a cycle $C$ in $G(\bm{\hat{x}})$. If $w_n \not \in V(C)$, then $C \subseteq G(\bm{x}^*)$ would hold. If $w_n \in V(C)$, then $P:=C-\{w_n u_1,w_nu_2\}$ would be a $u_1 -u_2$-path such that $P\subseteq G(\bm{x}^*)$ and $P \triangle \{w_j u_1, w_j u_2\}$ would contain a cycle in $G(\bm{x}^*)$, where $\triangle$ denotes the symmetric difference of two sets.

    In the second case, $\deg_{G(\bm{\hat{x}})} (u_1) \le 2$, $\deg_{G(\bm{x}^*)}(u_1) = 0$, $w_j u_1 \not \in E(G(\bm{x}^*))$. Assume for contradiction that there is a cycle $C$ in $G(\bm{\hat{x}})$. If $C$ does not contain any of the edges $w_j u_1, u_1 w_n, w_n u_2$, then $C \subseteq G(\bm{x}^*)$ would hold. Otherwise, since $\deg_{G(\bm{\hat{x}})} (u_1) \le 2$ and $\deg_{G(\bm{\hat{x}})} (w_n)= 2$, $C$ must contain all three edges. But then $P:= C - \{w_j u_1, u_1 w_n, w_n u_2\}$ is a $u_2 - w_j$ path in $G(\bm{x}^*)$, which together with $u_2 w_j$ would form a cycle in $G(\bm{x}^*)$. \hfill \Halmos

\section{Pseudocodes}\label{app:pseudocode}

\subsection{The branch-and-bound framework}
As already discussed in Section \ref{sec:intro}, any B\&B methods run some basic functions that can be highly customized. 
Algorithm \ref{alg:branch_and_bound} details a general B\&B framework for a minimization problem, but a similar argument can occur with a maximization one. 

\begin{algorithm}
\caption{Branch and Bound Algorithm. The steps denoted with $*$ must be changed when switching from minimization to maximization. }
\label{alg:branch_and_bound}
\footnotesize
\begin{algorithmic}[1]
    \State \textbf{Input:} Problem instance, a threshold that we wish to guarantee to our solution quality
    \State \textbf{Output:} High-quality solution
    \State Do any necessary preprocessing
    \State Initialize global lower bound (LLB) and global upper bound (LUB) (best feasible solution)
    \State Compute initial lower bound using a relaxation method
    \If{is integer}
        \textbf{return}
    \EndIf
    \State Initialize priority queue (heap) with root node
    \While{queue is not empty}
        \State Extract the most promising node from the queue
        \If{node is integral}
            \State Update LUB$^*$ if a better solution is found
            \State \textbf{continue}
        \EndIf
        \State Select a branching item/job
        \For{each possible branch (child node)}
            \State Apply feasibility check
            \State Compute new upper bound and lower bound
            \If{new lower bound $<$ LUB$^*$}
                \State Add new node to the queue
            \EndIf
        \EndFor

        \State Update LLB$^*$ as minimal remaining lower bound in queue
        \If{A relation between LUB, LLB and the threshold is satisfied}
            \State \textbf{return}
        \EndIf
    \EndWhile
\end{algorithmic}
\end{algorithm}

\subsection{A specific implementation}

Now we describe how Algorithm \ref{alg:branch_and_bound} is specified to obtain $A^{\text{knap}}_{\alpha}$. The other algorithms can be derived similarly.

Let Solve-LP denote the subroutine that on input $(\bm{C}, \bm{w}, \bm{p})$, returns the triple $(\bm{x}^*, \bm{x'}, j^*)$ as described in Proposition \ref{2_approx_knapsack} and Lemma \ref{knapsack_gap_and_critical_element}: 
$\bm{x^*}$ is the fractional optimum of the knapsack instance, $\bm{x}'$ is the best assignment among $\lfloor\bm{x}^*\rfloor$ and the critical elements, and $j^*$ is the most profitable critical element.

We describe the branch-and-bound algorithm $A^{\text{knap}}_{\alpha}$ in detail below. Each node $v$ will be identified by the unique sets $(I, E)$ with $I = \{(j_1,i_1), \ldots, (j_k, i_k)\}$ being the (item, knapsack) inclusions fixed so far, and $E=\{(e_1,m+1), \ldots, (e_l,m+1)\}$ being the excluded items.

\begin{algorithm}
\caption{$A^{\text{knap}}_{\alpha}$. The subroutine Solve-LP returns the triple of optimal fractional solution, its rounded-up integer solution and the most profitable critical item.}
\label{alg:knapsack}
\footnotesize
\begin{algorithmic}[1]
    \State \textbf{Input:} A knapsack instance $(\bm{C}, \bm{w}, \bm{p})$ with a fixed number of knapsacks.
    \State \textbf{Output:} An integer assignment whose profit is at least $\alpha$ times the optimum.
    \State $r := (\emptyset, \emptyset)$
    \State $(\bm{x}^* _r, \bm{x}_r ', j^* _r) := \text{Solve-LP}(\bm{C, \bm{w}, \bm{p}})$
    \State $UB(r):= \bm{p}\cdot \bm{x}^* _r$, $LB(r):= \bm{p}\cdot \bm{x}' _r$
    \State $HUB := U(r)$, $HLB := L(r)$
    \State $queue :=  \{r\}$
    \While{$\frac{HLB}{HUB} < \alpha$}
        \State $v := \arg\max \{U(node): node \in queue\}$
        \State $(I, E) \leftarrow v$, $\{(j_1, i_1, )\ldots, (j_k, i_k)\} \leftarrow I$, $\{(e_1, m+1), \ldots, (e_l, m+1)\} \leftarrow E$
        \State $queue := queue \setminus \{v\}$
        \State $S := [n]\setminus \left(\bigcup_{z=1}^k \{j_z\} \cup \bigcup_{z=1}^l \{e_z\}\right)$
        \State $\bm{w}_v := \bm{w}|_{S}$, $\bm{p}_v := \bm{p}|_{S}$
        \For{$i=1, \ldots, m$}
            \State $C^v_i := C_i - \sum\limits_{z: i_z = i} w_z$
        \EndFor
        \State $\bm{C}_v := (C^v _1, \ldots, C^v _m)$
        \State $(\bm{x^*}, \bm{x}', j^*):= \text{Solve-LP}(\bm{C}_v, \bm{w}_v, \bm{p}_v)$
        \For{$i=1, \ldots, m$}
            \State $v_i := (I\cup (j^*, i), E)$
        \EndFor
        \State $v_{m+1} := (I, E\cup (j^*, m+1))$
        \For{$i=1, \ldots, m$}
            \State $\bm{C}_{v_i} := (C^v _1, \ldots, C^v _{i-1}, C^v_i - w_{j^*}, C^v_{i+1}, \ldots, C^v_m)$
        \EndFor
        \State $\bm{C}_{v_{m+1}} := \bm{C}_v$, 
        \For{$i=1, \ldots, m+1$}
            \State $\bm{w}_{v_i} := \bm{w} |_{S \cup \{j^*\}}$, $\bm{p}_{v_i} := \bm{p} |_{S \cup \{j^*\}}$
            \State $(\bm{x}^* _{v_i}, \bm{x}' _{v_i}, j^*_{v_i}) := \text{Solve-LP}(\bm{C}_{v_i}, \bm{w}_{v_i}, \bm{p}_{v_i})$
            \State $SUB(v_i) := \bm{p}_{v_i} \cdot \bm{x}^* _{v_i}$, $SLB(v_i) := \bm{p}_{v_i} \cdot \bm{x}' _{v_i}$
            \State $UB(v_i) := SUB(v_i) + \sum\limits_{(j,t) \in I \cup \{(j^*, i)\}, t \ne m+1} p_j$
            \State $LB(v_i) := SLB(v_i) + \sum\limits_{(j,t) \in I \cup \{(j^*, i)\}, t \ne m+1} p_j$
            \If{$UB(v_i) > HLB$}
                \State $queue := queue \,\cup \{v_i\}$
            \EndIf
        \EndFor
        \State $HUB = \max\{UB(node): node \in queue\}$, $HLB = \max\{LB(node): node \in queue\}$
    \EndWhile
    \State $v := \arg\max \{UB(node): node \in queue\}$
    \State $(I,E) \leftarrow v$
    \State \textbf{return} $I \cup \{(j,i): (\bm{x}'_v)_{j,i}=1\}$.
\end{algorithmic}
\end{algorithm}

\end{APPENDICES}
\end{document}